\documentclass[12pt]{article}
\usepackage{amssymb, amsmath, amsthm}
\usepackage{thmtools}
\usepackage{mathtools}
\declaretheorem[style=definition]{example}
\usepackage{appendix}
\usepackage{bm}
\usepackage{bbm}
\usepackage{caption}
\usepackage{enumerate} 
\usepackage{graphicx}
\usepackage{natbib}
\usepackage{setspace}
\usepackage{hyperref} 
\usepackage{color}
\usepackage{tikz}
\usepackage{mathrsfs}
\usetikzlibrary{snakes}
\usepackage{subfigure}
\numberwithin{equation}{section}
\usepackage[top=1.25in, bottom=1.25in, left=1.25in, right=1.25in]{geometry}
\DeclareMathOperator*{\argmax}{arg\,max}

\DeclareMathOperator*{\supp}{supp}
\DeclareMathOperator{\sgn}{sgn}
\newtheorem{assm}{Assumption}
\newtheorem{lem}{Lemma}

\newtheorem{thm}{Theorem}
\newtheorem*{thm*}{Theorem}
\newtheorem{cor}{Corollary}

\newtheorem{defn}{Definition}
\newtheorem{rem}{Remark}

\begin{document}
\renewcommand{\thefootnote}{\fnsymbol{footnote}}
\renewcommand\thmcontinues[1]{Continued}
\title{The Politics of Attention}
\author{
Lin Hu\footnote{Research School of Finance, Actuarial Studies and Statistics, Australian National University. lin.hu@anu.edu.au. }
\and Anqi Li\footnote{Department of Economics, Washington University in St. Louis. anqili@wustl.edu. We thank Steve Callander, Matt Gentzkow, Kun Li, Ilya Segal, Chris Shannon, Ken Shotts, Takuro Yamashita, as well as the seminar audience at APEN 2018, Bocconi, Columbia, Toulouse, UC Davis, University of Konstanz, University of Warwick and University of Zurich for comments and suggestions. All errors are our own.}
}
\date{First Draft: July 2018\\This Draft: January 2019 }
\maketitle

\begin{abstract}
We develop an equilibrium theory of attention and politics. In an electoral competition model where candidates have varying policy preferences, we examine what kinds of political behaviors capture voters' limited attention and how this concern affects the overall political outcomes. Following the seminal works of Downs (1957) and Sims (1998), we assume that voters are rationally inattentive and can process information about the policies at a cost proportional to entropy reduction. The main finding is an equilibrium phenomenon called  \emph{attention-} and \emph{media-driven extremism}, namely as we increase the attention cost or garble the news technology, a truncated set of the equilibria captures voters' attention through enlarging the policy differentials between the varying types of the candidates. We supplement our analysis with historical accounts, and discuss its relevance in the new era featured with greater media choices and distractions, as well as the rise of partisan media and fake news.

%We develop an equilibrium theory of attention and politics. In an electoral competition model where candidates have varying policy preferences, we assume that voters are rationally inattentive and can process information about policies at a cost proportional to entropy reduction. As we increase the attention cost or garble the news technology, a truncated set of the equilibria captures voters’ attention through enlarging the policy differentials between the varying types of candidates. We supplement our analysis with historical accounts, and discuss its relevance in the new era featured with greater media choices and the rise of partisan media and fake news.

\bigskip 

\bigskip 

%\noindent Keywords: rational inattention; politics; media.
%
%\bigskip 

\end{abstract} 
\renewcommand{\thefootnote}{\arabic{footnote}}
\pagebreak

\section{Introduction}\label{sec_intro}
The 1790's witnessed the separation of two important figures in the U.S. history: Alexander Hamilton and James Madison. Once closed allies who coauthored the Federalist papers and shared huge successes in justifying and marketing the U.S. Constitution, the two figures started to disagree about Hamilton's economic policies,  which emphasized finance, manufacturing and trade over agriculture. To fight their war and in particular, to arouse and attract public attention,  the two of them  founded political parties and partisan newspapers with extreme and exaggerated positions. To give a vivid example, while Madison himself believed that a healthy country should excel in all areas, he urged that people look inwards, to the center of the country, to farmers, and that they go back to the values that made America great, namely low taxes, agriculture and less trade  (\cite{feldman}). 

In this paper, we ask what kinds of political behaviors capture voters' limited attention and how this concern affects the overall political outcomes. The premise of our analysis, as asserted by \cite{downs},  is that attention is a scarce resource and its utilization should be governed by meticulous cost-benefit analysis.\footnote{As Anthony Downs famously postulated in \emph{An Economic Theory of Democracy} (1957): ``In our model, as in real life, political decisions are made when uncertainty exists and information is obtainable only at a cost. Thus a basic step towards understanding politics is analysis of the economics of being informed, i.e., the rational utilization of scarce resources to obtain data for decision-making.'' } While descriptive studies of rational inattentive voters abound, theoretical analysis is lacking, leaving important questions unanswered. For example, what is the equilibrium effect of limited attention on policies and voter behaviors? Can the changing market conditions heatedly debated by the public--such as greater media choices and distractions, the rise of partisan media and fake news, etc.--partially explain the recent political landscape through the channel of limited voter attention?

We seek to understand the above described questions in a spatial model of electoral competition. Inspired by the existing studies on intra-party politics, we assume that each of the two candidates can have varying  preferences for policies, privately known to themselves, e.g., conservative or liberal left. In particular, we allow such preferences to depend on the candidate's winning status, thus ensuring that the equilibrium policies of the varying types of the candidates can indeed be uncertain from the voter's perspective. 

In our model, as in \cite{gentzkow2} and \cite{pewmediahabit}, attending to politics means taking actions that help reduce policy uncertainties and make voting decisions. Doing so incurs real and opportunity costs, as the time and effort required for the processing and absorption of political contents, holding debates and deliberations that facilitate thinking, etc. could be spent elsewhere, such as work, leisure and entertainment. Following \cite{sims} and \cite{sims1}, we assume that attention cost is proportional to the mutual information between policies and voting decisions. It is scaled by a parameter called the \emph{marginal attention cost}, which represents cost shifters such as exposures to distractions and media choices, the competition between firms for consumer eyeballs, etc.  (\cite{cabletv}; \cite{prior}; \cite{teixeira}; \cite{dunaway}; \cite{perez}). 

Our choice of the attention function seems well-suited for today's information-rich world, where a larger body of political information is accessible to ordinary citizens than ever before. Attention becomes a scarce resource as foreseen by \cite{simon}, and part of it is paid to politics through our regular exposures to a selective yet wide range of sources (\cite{gentzkow2}; \cite{iyengar}). For the reasons discussed in Section \ref{sec_cost}, we model attention as a selective communication channel \`{a} la \cite{shannon}, whose long-run operating cost is well known to be the mutual information between the source data and the voting decision. The material of Appendix  \ref{sec_voter} provides further vindications, documenting interesting consequences of rational, flexible attention allocation, such as selective exposure, confirmatory biases and seeking big occasional surprises. 

Our main finding is an equilibrium phenomenon called \emph{attention-driven extremism}, namely as we increase the marginal attention cost, a truncated subset of the equilibria captures voters' attention through enlarging the policy differences between the varying types of the candidates. As demonstrated in Appendices \ref{sec_commitment} and \ref{sec_multi}, this result holds true even if (1) policies are multi-dimensional, suggesting that our analysis is not confined to left-right politics;  or if (2) candidates can only partially honor their policy proposals, which can be regarded as their promises made during the campaign. 

It is worth distinguishing our result from the following best-response argument: as attention becomes costly, voters do not bother to notice the small policy differences between the varying types of the candidates, hence the latter should push themselves apart as a means of retaining attention.  This argument is incomplete, because  it ignores the equilibrium effects that attention and policies can potentially exert on each other. Theorem \ref{thm_main} goes one step further, showing--using robust properties of the attention function--that under general assumptions about the environment, voter characteristics, such as marginal attention cost, have no effect on candidates' winning probabilities. Thus in  equilibrium, marginal attention cost affects only attention strategies but not policies, and increases in this parameter reduce the set of equilibria that arouses and attracts attention and raise the needed degree of policy extremism for achieving this goal. 

In our baseline model, as in many existing electoral competition models, candidates care about voters only through the channel of winning probability. In  Appendix  \ref{sec_selection}, we use  visibility as an equilibrium selection device, assuming that the dissemination of political information involves entities that profit from voters' eyeballs, e.g., revenue-maximizing content providers. Assuming that these entities incur no loss in equilibrium, we obtain a stronger result: increases in the marginal attention cost leave us with a truncated set of policy distributions that exhibits a greater degree of extremism than ever. We discuss supporting evidence in Section \ref{sec_evidence}. 

In Section \ref{sec_noisy}, we assume that news about candidates is itself a random variable, and call its probability distribution the \emph{news technology}. Voters pay costly attention to news before casting votes. Recently, economists, journalists and political scientists have voiced concerns for the rise of partisan media and fake news (\cite{levendusky}; \cite{handbooktheory}; \cite{handbookempirical}; \cite{gentzkowfakenews}). Following \cite{handbooktheory}, we model the above described changes as Blackwell (1953)'s garbling of the news technology and examine their equilibrium effects through the channel of limited voter attention.

The results appeared in Theorems \ref{thm_noisy_policy} and \ref{thm_noisy_attention} are more delicate than their equivalents in the baseline model. On the one hand, the irrelevance of marginal attention cost to equilibrium policies remains valid. On the other hand, garbling affects both policies and attention strategies, and the overall effect can be subtle. The lesson is twofold. First, any equilibrium that grabs voters' attention after garbling must exhibit greater policy differentials than those that fail to do so prior to garbling -- a phenomenon we term as \emph{media-driven extremism.} Second, in order to back out the extent of garbling from real-world data, much effort should be spent on the identification of shocks that affect only attention strategies but not  policies, and the attention cost shifters as discussed in Section \ref{sec_cost} seem to serve this purpose well. 

At a high level, our theory formalizes the role of rationally inattentive voters in the overall political landscape, and calls for its rigorous analysis through the equilibrium channels identified above. Without it, our understandings of certain events in human history will be deemed incomplete. A famous example dates back to the 1950s, when the conformity in national politics of the United States--hallmarked by Eisenhower's embrace of New Deal and the conservative Democrats' seizure of the congress--created a sentiment among voters that was best characterized apathy and indifference. Sensing the pressure from the other side, local candidates exercised increasing caution over time, restricting policy choices to those that pushed the varying types of themselves far apart and forgoing the remainder that exhibited miniature differences. According to \cite{campbell}, visibility seemed to be the primary concern, since ``In the electoral as a whole where the level of attention is so low, what the public is exposed to must be highly visible--even stark--if it is to have an impact on opinion.''

Our result sheds light on how we should interpret voter data. As discussed in Appendix \ref{sec_nonmonotone}, a consequence of rational, flexible attention allocation is that the average propensity one supports the candidate from his own ideological camp can vary non-monotonically with the marginal attention cost. Once we realize that opinions and votes are the result of voter paying limited attention to their political surroundings, two takeaways are immediate. First, one should not equate measured behaviors with intrinsic preferences and should instead tease them apart using the methods proposed and reviewed in \cite{caplindean}. Second, in light of the changing market conditions that affect the attention cost, one should not be too surprised that the evidence on mass polarization as documented in \cite{fiorina} and \cite{gentzkow} is at best mixed. 

\subsection{Related Literature}

\paragraph{Probabilistic voting models} Our candidates have random policy preferences, and they  know the median voter's position. Thus the main difference between the probabilistic voting models pioneered by \cite{wittman} and \cite{calvert1} is twofold.\footnote{Variations of classical models abound. See, among others, \cite{duggan} for a comprehensive survey of the theory literature, as well as the textbook of \cite{persson} for applications to important institutions such as public spending and redistribution. } First, our notion of policy extremism, defined by the policy gap between the varying types of the same candidate, should not be confused with the \emph{policy divergence} (between different candidates) that probabilistic voting models aim to produce.  Second, our voters act randomly because of limited attention, but they together impose no aggregate uncertainty on candidates. Thus, factors that affect the randomness of individual votes, such as marginal attention cost, are irrelevant in the determination of equilibrium policies, and they should not be regarded as sources of the aggregate-level voter uncertainty that candidates face in probabilistic voting models. 

\paragraph{Electoral competition with costly information acquisition} 
Several recent papers introduce costly information acquisition into electoral competition models.\footnote{A common feature of these works is that voters can take actions that affect the uncertainties they face. This possibility is ruled out by \cite{snyder}, \cite{glaeseretal},  \cite{gulpesendorfer} and \cite{aragones2}, in which voters treat the partial observability of the fundamentals as exogenously given.  %\cite{lizzeri}, 

One should not confuse such partial observability with the ``strategic ambiguity'' of the political messages as studied in \cite{shepsle}, \cite{alesina}, \cite{callanderwilson} and the follow-up works. In the former case, policies are vague because voters lack the capacity to fully absorb their contents. In the latter case, voters fully absorb the messages announced by candidates, but messages are only partially revealing of the fundamentals due to reasons such as risk preferences, asymmetric information and behavioral considerations.}  \cite{matejka} examines a probabilistic voting model with office-motivated candidates. Unlike our model where policy uncertainty arises from policy preferences, there voters can reduce the variances of normally distributed random variables equal to the policies plus exogenous shocks, and the cost of variance reduction can depend on the candidates' informational attributes such as transparency and media coverage. As noted by these authors, their model produces divergence, apart from other illuminating results, if candidates differ by information attributes. Our notion of policy extremism differs, and our analysis assumes symmetry. 

\cite{yuksel} allows the voters of probabilistic voting models to partially observe their aggregate preference shock. In the case of single-issued policies, improvements in the information technology is shown to produce policy divergence through enlarging the information gap between players. 

 \cite{prato} develops a model of electoral communication in which paying attention increases the chance that the voter discovers the winning candidate's plan once in office. Surprisingly, too much attention can be a curse, because it may lead incompetent candidates to make inefficient policy choices. 

\paragraph{Incomplete information about candidates' policy preferences} A key element of our model is a static incomplete information game in which candidates are privately informed of their policy preferences, and they care about voters through the channel of winning probability. Private information is modeled as a type that takes finitely many values, each specifying the candidate's utility functions in cases where he wins the election and not. Prior studies such as \cite{banks}, \cite{kartik}, \cite{callander}, \cite{callanderwilkie} tend to work with specific preference and information structures, and some of them can be nested by our framework after careful modifications (more on this later). Nevertheless, we view our games as distinct ones along other important dimensions, such as whether candidates can take actions besides setting policies or not. 

 \paragraph{Rational inattention}  
 Early development of the literature on rational inattention (hereinafter, RI) pioneered by \cite{sims}, \cite{sims1}, \cite{riprice} and \cite{woodford1} sought to explain the stickiness of macroeconomic variables by entropic information costs. Recently, a great deal of effort by \cite{riconsumer}, \cite{yang0}, \cite{yang}, \cite{ravid}, \cite{martin}, etc., has been devoted to understanding the impact of RI for strategic interactions. A common thread that runs through these works is the flexibility that pertains to RI attention allocation. In our case, such flexibility paves the way to sharp equilibrium characterizations, as it ensures that the monotonicity properties of voters' preferences will pass seamlessly along to optimal attention strategies. The construction combines the results of  \cite{mckay} and \cite{yang} with techniques in monotone comparative statics.

\paragraph{Media} Our approach to studying media bias is borrowed from \cite{handbooktheory}, and our investigation of its joint effect with limited voter attention is absent from the existing political models such as \cite{stromberg}, \cite{dugganmartinelli} and  \cite{gulpesendorfermedia}. By modeling media as an exogenous technology, we abstract away from the issue of self-interested media, and we refer the reader to the handbook chapters of \cite{handbooktheory}, \cite{handbookempirical} and \cite{handbookprat} for thorough reviews of the existing literature. Embedding self-interested media into political models with rationally inattentive voters is the subject matter of our companion work \cite{fragmentation}. 

\bigskip

The remainder of this paper proceeds as follows: Section \ref{sec_baseline} introduces the model setup; Section \ref{sec_eqm} conducts equilibrium analysis; Section \ref{sec_noisy} investigates an extension of the baseline model; Section \ref{sec_conclusion} concludes. Omitted proofs and supplementary materials can be found in Appendices \ref{sec_voter}--\ref{sec_online}.

\section{Baseline Model}\label{sec_baseline}
\subsection{Setup}\label{sec_setup}
\paragraph{Primitives} There is a unit mass of infinitesimal voters and two candidates named $\alpha$ and $\beta$. Each player has a \emph{type} that affects his valuations of the policies in $\Theta=[-1,1]$, and the distribution of types is independent across players. In particular, voters' types follow a continuous distribution $F$ with full support $\Theta$ and zero median, whereas each candidate $c$'s type is drawn from a vector space $T$ according to a finite distribution $P_c$, $c=\alpha,\beta$. In what follows, let $t$ and $\widetilde{t}_c$ denote the type of a typical voter and candidate $c$, respectively, and divide the voters into \emph{pro-$\alpha$} voters, the \emph{median voter} and \emph{pro-$\beta$ voters}  based on whether their types belong to $\Theta_{\alpha}=[-1,0)$, $\Theta_0=\{0\}$ or $\Theta_{\beta}=(0,1]$. 

%(e.g., left and right, pro- and anti-animal rights)

 Each candidate $c$ has access to the policies in a \emph{finite subset} $A_c$ of $\Theta_c$.\footnote{This assumption is often stronger than needed, and its role will be commented on after the statement of Assumption \ref{assm_u}.} In the case where candidate $c$ assumes office and implements policy $a$,  the utility of type $t$ voter, the winning candidate and the losing candidate is $u(a,t)$, $u_+(a, t_c)$ and $u_-(a, t_{-c})$, respectively. The functions $u: \Theta \times \Theta \rightarrow \mathbb{R}$, $u_+: \Theta \times T \rightarrow \mathbb{R}$ and $u_-: \Theta \times T \rightarrow \mathbb{R}$ are continuous in the first argument, and their properties will be further discussed in Section \ref{sec_eqm}. For now, it is worth noting that our candidates can have  office and policy motivations, both in general forms, as their utilities can depend on who is in charge and which policy is being implemented. 
 
\paragraph{Election game} We add three ingredients to Downs' (1957) model: random policy preferences, limited voter attention and noisy news. Time evolves as follows:

\begin{enumerate}
\item Nature draws types for players;

\item candidates observe their own types and simultaneously propose policies denoted by $a_c$, $c=\alpha, \beta$;

\item the press releases news $\bm{\omega}$ about the \emph{policy profile} ${\bf{a}}=\left(a_{\alpha}, a_{\beta}\right)$; 

\item voters \emph{attend to politics} and vote for one of the candidates;

\item winner is determined by simple majority rule with even tie-breaking and implements his policy proposal in Stage 2.
\end{enumerate}

\paragraph{Candidate's strategy}  Each candidate $c$'s strategy is $\sigma_c: \supp\left(P_c\right) \rightarrow \Delta(A_c)$, and the profile of candidate strategies is $\sigma=(\sigma_{\alpha}, \sigma_{\beta})$. The policy profile $\widetilde{\bf{a}}$ induced by $\sigma$ is a non-degenerate random variable if $\vert \supp(\sigma) \rvert \geq 2$, in which case we simply say that $\sigma$ \emph{is non-degenerate}. 

\paragraph{Voter's problem}  We make two assumptions about voters.  First, voting is \emph{expressive} as in most election models, meaning that each voter chooses his preferred candidate based on the information available. 

 Second, attending to politics helps voters reduce policy uncertainty and decide which candidate to choose, and the cost is proportional to the mutual information between the news and the decision. We motivated this assumption in Section \ref{sec_intro} and will continue to do so in Section \ref{sec_cost}. 

The baseline model deals with the benchmark case where news coincides with the policy profile itself, i.e., $\bm{\omega}={\bf{a}}$. By Lemma 1 of \cite{mckay}, we can summarize the decisions of an arbitrary voter $t$ by an \emph{attention strategy} $m_t: A_{\alpha} \times A_{\beta} \rightarrow \left[0,1\right]$, where $m_t({\bf{a}})$ represents the probability that voter $t$ chooses candidate $\beta$ under policy profile $\bf{a}$. Let 
\begin{equation*}
v\left({\bf{a}}, t\right)=u\left(a_{\beta}, t\right)-u\left(a_{\alpha}, t\right)
\end{equation*}
be the voter's differential utility from choosing candidate $\beta$ over candidate $\alpha$ under policy profile $\bf{a}$, and let 
\begin{equation*}
V_t\left(m_t,\sigma\right)=\sum_{{\bf{a}} \in \supp(\sigma)} m_t\left({\bf{a}}\right) v\left({\bf{a}}, t\right) \sigma({\bf{a}}) 
\end{equation*}
be the expected differential utility under any given strategy profile $\left(m_t, \sigma\right)$. The voter's expected payoff is then 
\begin{equation*}
V_t\left(m_t, \sigma\right)-\mu \cdot I\left(m_t,\sigma\right), 
\end{equation*}
where the second term of the above expression is the attention cost. 

\paragraph{Candidate's payoff}  For any policy profile $\bf{a}$, define 
\begin{equation}\label{eqn_wp}
w\left({\bf{a}}\right)=\begin{cases}
0 & \text{ if } \int m_t\left({\bf{a}}\right) dF(t)<\frac{1}{2},  \\ 
\frac{1}{2} & \text{ if } \int m_t\left({\bf{a}}\right) dF(t)=\frac{1}{2}, \\
1 & \text{ if } \int m_t\left({\bf{a}}\right) dF(t)>\frac{1}{2}, 
\end{cases}
\end{equation}
and let $w_{\alpha}({\bf{a}})=1-w({\bf{a}})$ and $w_{\beta}({\bf{a}})=w({\bf{a}})$ be the winning probability of candidate $\alpha$ and $\beta$, respectively. Let $m=\left(m_t\right)_{t \in \Theta}$ denote the attention strategies of all voters. Under any strategy profile $\left(m, \sigma\right)$, the expected payoff of candidate $c$ is 
\begin{equation}\label{eqn_candidate}
V_c(m, \sigma)= \mathbb{E}_{m, \sigma}\left[w_c\left(\widetilde{\bf{a}}\right)  u_+\left(\widetilde{a}_c, \widetilde{t}_c\right) + \left(1-w_c\left(\widetilde{\bf{a}}\right)\right) u_-\left(\widetilde{a}_{-c}, \widetilde{t}_c\right) \right].
\end{equation}

\paragraph{Equilibrium} A strategy profile $\left(m^*, \sigma^*\right)$ constitutes a \emph{Bayes Nash equilibrium} of the election game if 
\begin{enumerate}
\item each $m_t^*$ maximizes voter $t$'s expected payoff, taking $\sigma^*$ as given, i.e., $\forall t$, 
\begin{equation*}
m_t^* \in \argmax_{m_t: A_{\alpha} \times A_{\beta}\rightarrow [0,1]} V_t\left(m_t, \sigma^*\right) - \mu \cdot I\left(m_t, \sigma^*\right); 
\end{equation*}

\item each $\sigma_c^*$ maximizes candidate $c$'s expected payoff, taking $m^*$ and $\sigma_{-c}^*$ as given, i.e., $\forall c$, 
\begin{equation*}
 \sigma_c^* \in \argmax_{\sigma_c} V_c\left(m^*,\sigma_c, \sigma_{-c}^*\right). 
\end{equation*}
\end{enumerate}

\begin{rem}
In the case where not all feasible policy profiles arise on the equilibrium path, we adopt the standard refinement in the RI literature that any off-equilibrium path policy profile can be observed by voters without errors,\footnote{Such best-response arises in the limit as we restrict candidates to adopting totally mixed strategies and take the probability of trembling to zero. See, e.g., \cite{ravid} for further discussions. } i.e., $\forall {\bf{a}} \notin \supp\left(\sigma^*\right)$, 
\begin{equation}\label{eqn_perfect}
m_t^*\left({\bf{a}}\right)=\begin{cases}
1 & \text{ if } v\left({\bf{a}}, t\right)>0,\\
0 & \text{ else}. 
\end{cases}
\end{equation}
As demonstrated later, this assumption plays a minimal role in our qualitative predictions, and it will be disposed of in Section \ref{sec_noisy} where news is assumed to be a noisy signal drawn from a full-support distribution. 
\end{rem}

\subsection{Rational Inattentive Voters: Part I}\label{sec_cost}
 Downs' (1957) thesis has two main elements: attention is costly and attention allocation is rational. Both elements are grounded in reality.

\paragraph{Attention is costly}  Throughout the analysis, attending to politics means taking actions that help reduce policy uncertainties and make voting decisions. Examples include: the processing and absorption of political contents, conversing with colleagues, family members and friends, and holding debates and deliberations that facilitate thinking  (\cite{gentzkow2}; \cite{pewmediahabit}). 

In today's information-rich world, our attention capacity is rather limited. Among other things, greater exposures to distractions and media choices seem to have significant behavioral impacts, as they have shifted our attention from politics to entertainment, and have lowered the levels of engagement with political contents among mobile device users (\cite{cabletv}; \cite{prior}; \cite{teixeira};  \cite{dunaway}; \cite{perez}).\footnote{For example, \cite{teixeira} estimates that the cost of consumer attention has risen by seven to nine times in the past two decades. \cite{perez} reports that mobile apps for gaming, messaging, music and social alone seize up to three hours of U.S. users' time per day as of 2016.}

\paragraph{Attention as a selective communication channel} The idea that voters focus on the part of the politics they care most about is not new.  A famous example dates back to the 1950s, when Eisenhower's differing attitudes towards agriculture and industry led farmers and laborers to tell different stories about his leadership (\cite{campbell}). Recent studies of  \cite{gentzkow2} and \cite{iyengar} find that people obtain political news through regular exposures to a selective yet wide range of sources. This suggests that we can model attention as a communication channel  \`{a} la \cite{shannon}.

Here is how things work. Imagine there is a myriad of sources, each conveys a bit of information and together paint a holistic picture of the policies. Every time before making a decision, the voter visits a number of sources sequentially at a constant marginal cost, and the problem is repeated for many times. To save costs, the voter optimizes over the sources based on the decision he wishes to make. For example, if the decision is to vote unconditionally for one of the candidates, then no consultation is needed. If the decision varies significantly with certain aspects of the policies but not others, then priority should be given to the sources that tell the most relevant stories. On average, the minimum number of the visited sources is approximately equal to the mutual information between the source data and the decision, and this is the famous result of  \cite{shannon}.

Following \cite{sims} and \cite{sims1}, we assume the following attention cost function throughout the analysis: 
\[\underbrace{\mu}_\text{marginal attention cost} \cdot \underbrace{I\left(m_t, \sigma\right)}_\text{mutual information}. \]
The parameter $\mu>0$ is called the \emph{marginal attention cost} and is used to  capture the aforementioned cost shifters. In principle, we could allow this parameter to vary with the voter's   predisposition, but we choose not to do this, because nothing changes except when we conduct in-depth analysis of voter behaviors in Appendix \ref{sec_voter}. 

The term $I\left(m_t, \sigma\right)$ is the mutual information between the policies and the voting decision, jointly distributed according to $(m_t, \sigma)$. Intuitively, if a voter pays no attention to politics, then the uncertainty he faces can be captured by the entropy of the policy profile: 
\[
H\left(\sigma\right)=-\sum_{{\bf{a}} \in \supp(\sigma)} \sigma\left({\bf{a}}\right)\log \sigma\left({\bf{a}}\right). 
\]
If he gathers information as above, then he has some partial idea about the policies by the decision-making time. For any given attention strategy $m_t$, let 
\[
\overline{m}_{t,s}=\begin{cases}
\mathbb{E}_{\sigma}\left[m_t\left(\widetilde{\bf{a}}\right)\right]& \text{ if } s=\beta,\\
\mathbb{E}_{\sigma}\left[1-m_t\left(\widetilde{\bf{a}}\right)\right] & \text{ if } s= \alpha,
\end{cases}
\]
be the marginal probability that the decision is $s=\alpha, \beta$. The residual uncertainty conditional on the decision being $s$ is then 
\[
H\left(\sigma_t\left(\cdot \mid s\right)\right)=-\sum_{{\bf{a}} \in \supp(\sigma)} \sigma_t\left({\bf{a}} \mid s\right) \log\sigma_t\left({\bf{a}} \mid s\right), 
\]
where 
\[
\sigma_t\left({\bf{a}} \mid s\right)=\begin{cases} \displaystyle \frac{m_t\left({\bf{a}}\right) \sigma\left({\bf{a}}\right)}{\overline{m}_{t,\beta}} &\text{ if } s =\beta,\\
 \displaystyle \frac{\left(1-m_t({\bf{a}})\right) \sigma\left({\bf{a}}\right)}{\overline{m}_{t,\alpha}} &\text{ if } s =\alpha,
\end{cases}
\]
is the conditional probability that the policy profile is $\bf{a}$, whenever the Bayes' rule applies. The reduction in entropy, also termed mutual information: 
\begin{equation*}
I\left(m_t, \sigma\right)=H\left(\sigma\right)-\sum_{s=\alpha, \beta} \overline{m}_{t,s} H\left(\sigma_t\left(\cdot \mid s\right)\right)
\end{equation*}
captures the amount of the time and effort one spends on information processing and digestion. This term increases as the decision becomes more informed, equaling zero if the latter is invariant with the policies and $H(\sigma)$ if it is to always choose one's preferred candidate. 

\paragraph{Optimal attention strategy} The next lemma gives characterizations of optimal attention strategies: 

\begin{lem}\label{lem_foc}
For any given profile $\sigma$ of candidate strategies, the optimal attention strategy of any voter $t$ uniquely exists. Let $m_t: A_{\alpha} \times A_{\beta} \rightarrow [0,1]$ be as such, and let \[\overline{m}_t=\mathbb{E}_{\sigma}\left[m_t\left(\widetilde{\bf{a}}\right)\right]\] be the average probability that voter $t$ chooses candidate $\beta$ under $\sigma$. 
Then, 
\begin{enumerate}[(i)]
\item if $\mathbb{E}_{\sigma}\left[\exp\left(v\left(\widetilde{\bf{a}}, t\right)/\mu \right) \right] < 1$, then $\overline{m}_t=0$;
\item if $\mathbb{E}_{\sigma}\left[\exp\left(-v\left(\widetilde{\bf{a}}, t\right)/\mu\right) \right] <1$, then $\overline{m}_t=1$;
\item otherwise $\overline{m}_t \in \left(0,1\right)$ and for any ${\bf{a}} \in \supp(\sigma)
$: 
\[m_t\left({\bf{a}}\right)=\frac{\Lambda_t\exp\left(\frac{v\left({\bf{a}},t\right)}{\mu}\right)}{\Lambda_t\exp\left(\frac{v\left({\bf{a}},t\right)}{\mu}\right)+ 1}, \] 
where
\[\Lambda_t=\frac{\overline{m}_t}{1-\overline{m}_t}\]
is the likelihood that voter $t$ chooses candidate $\beta$ over candidate $\alpha$ under $\sigma$; 
\item in Part (iii), $\overline{m}_t$ is strictly increasing in $t$ if the function $u$ has strict increasing differences in $\left(a,t\right)$. 
\end{enumerate}
\end{lem}

\begin{proof}
Parts (i) - (iii) are immediate from Theorem 1 of \cite{mckay} and Proposition 2 of \cite{yang}. The proof of Part (iv) can be found in Appendix \ref{sec_proof_cost}. 
\end{proof}

According to Lemma \ref{lem_foc}, the optimal attention strategy can take two forms, depending on whether the gain from attending to politics justifies the cost or not. In Parts (i) and (ii), the voter clearly prefers one candidate over another and acts deterministically without paying attention. In Part (iii), the voter has no clear preference a priori, and the result of him paying attention is a random decision that follows a shifted logit rule.

A closer look at the solution reveals interesting patterns. First, $m_t\left({\bf{a}}\right)$ depends on ${\bf{a}}$ \emph{only} through $v\left({\bf{a}},t\right)$, because distinguishing policy profiles of the same $v$ value does no good to decision-making but incurs redundant costs. 

\bigskip 

\begin{figure}[!h]
    \centering
     \includegraphics[height=8cm]{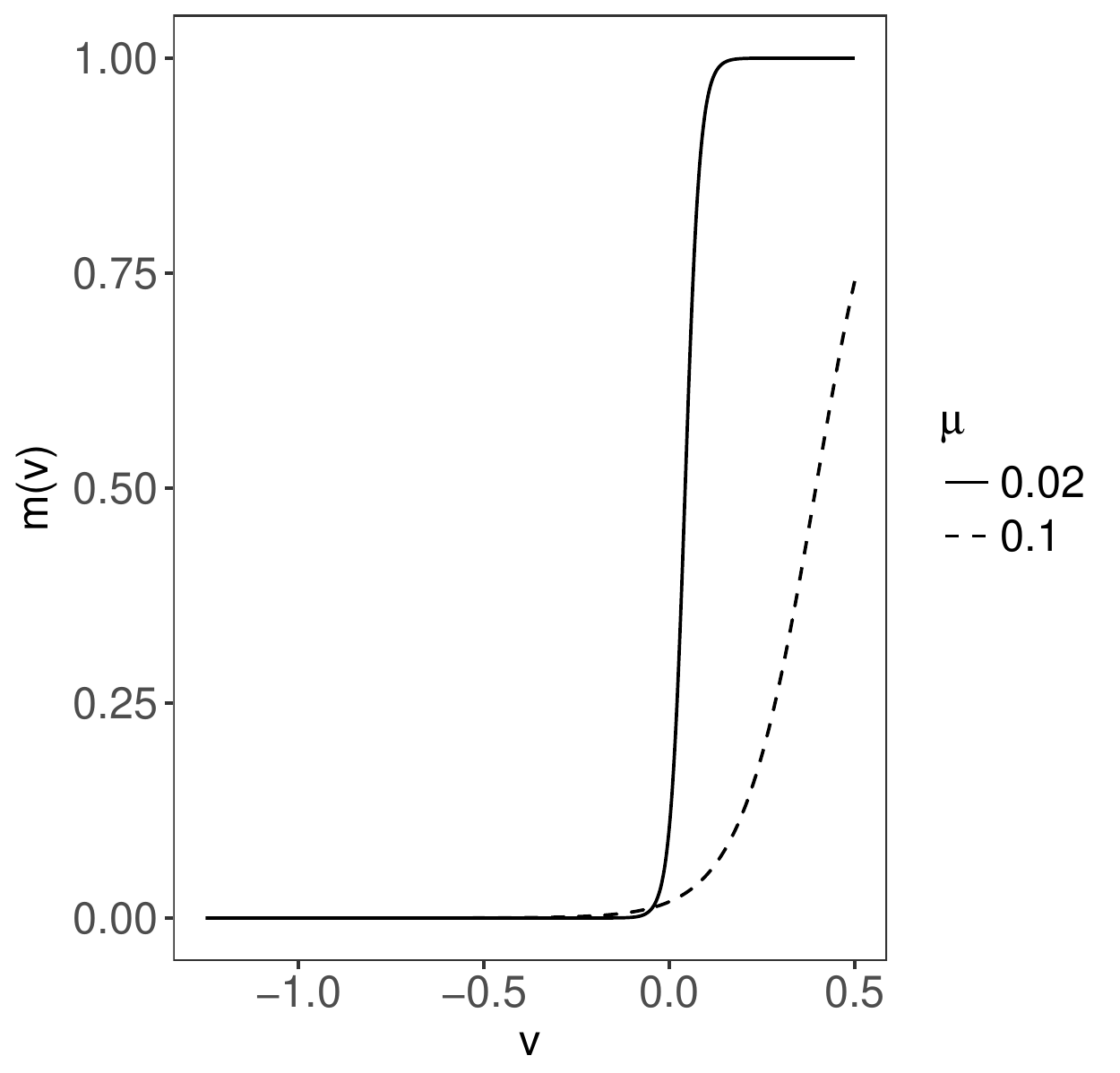}
    \caption{Plot optimal attention strategy against $v\left({\bf{a}},t\right)$ for $t=-.25$: $a_c$ is uniformly distributed on $\Theta_c$, $c=\alpha, \beta$, and $u(a,t)=-|t-a|$.}\label{figure_m}
\end{figure} 

 Second, the optimal attention strategy, when formulated as a function of $v$, is \emph{increasing} in its argument, hence one should choose a candidate more often as the gain increases. This monotonicity property, albeit a simple one, plays an important role in the upcoming analysis.

Third, optimal attention allocation is selective. While knowing the exact policy positions wouldn't hurt had attention been costless, such detailed information is rationally ignored even in light of an arbitrarily small attention cost. Intead, the focus becomes on which policy position is more preferable than the other, and the corresponding attention strategy resembles a step function that jumps from zero to one around $v=0$ as depicted in Figure \ref{figure_m}.\footnote{In the RI literature, ``focusing on event $E$'' means clearly distinguishing whether $E$ is true or not and varying the decision probability accordingly. This is different than the notion of focus as studied in \cite{nunnari}, which means making salient features even more salient.}  As attention cost increases, voters no longer  differentiate the varying positions as much, and the attention strategy flattens out accordingly. 

In Appendix \ref{sec_voter}, we continue the investigation of selective attention allocation, documenting interesting consequences such as (1) selective exposure, confirmatory biases and seeking big occasional surprises; as well as (2) the non-monotone variations of voter behaviors with personal characteristics such as political predisposition and marginal attention cost. We conclude with discussions of how we should interpret voter data and evaluate the effects of changing market conditions on mass polarization.

\section{Equilibrium Analysis}\label{sec_eqm}
In this section, we assume that the environment is symmetric around the median voter and focus on the symmetric equilibria of the election game:

\begin{assm}\label{assm_symmetry}
The environment satisfies the following properties:
\begin{enumerate}[(i)]
\item  $\widetilde{u}(a,t)=\widetilde{u}(-a,-t)$ for all $\widetilde{u} \in \left\{u, u_+, u_-\right\}$ and $(a,t) \in \mathrm{Dom}\left(\widetilde{u}\right)$;
\item $F(t)=1-F(-t)$ for all $t \in \Theta$ and $P_{\beta}(t)=P_{\alpha}(-t)$ for all $t \in T$.%\footnote{Since $T_{\alpha}$, $T_{\beta}$ are finite subsets of the vector space $V$, we can be assured that $-t$ is well-defined for every $t \in T_{\beta}$. }
\end{enumerate}
\end{assm}

\begin{defn}
A strategy profile $\left(m, \sigma\right)$ is \emph{symmetric around the median voter} if
\begin{enumerate}[(i)]
\item  $\sigma_{\beta} \left(a\mid t\right)=\sigma_{\alpha}\left(-a \mid -t\right)$ for all $a \in A_{\beta}$ and $t \in \supp\left(P_{\beta}\right)$; 
\item  $m_t\left(-a,a'\right)=1-m_{-t}\left(-a',a\right)$ for all $\left(-a,a'\right)\in A_{\alpha}\times A_{\beta}$ and $t \in \Theta$. 
\end{enumerate}
\end{defn}

The analysis makes the following assumptions about voters' utility function: 
\begin{assm}\label{assm_u}
The function $u: \Theta \times \Theta \rightarrow \mathbb{R}$ is continuous and satisfies the following properties: 
\begin{enumerate}[(i)]
\item  $u\left(\cdot, t\right)$ is strictly increasing on $[-1,t]$ and is strictly decreasing on $[t,1]$ for all $t$; 
\item $u(\cdot,t)$ is concave for all $t$;
\item $u(a',t)-u(a,t)$ is strictly increasing in $t$ for all $a'>a$; 
\item there exists $\kappa>0$ such that for all $t>0$, \[\min_{{\bf{a}} \in A_{\alpha} \times A_{\beta}} v\left({\bf{a}},t\right)-v\left({\bf{a}},0\right) >\kappa t.\]
\end{enumerate}
\end{assm}

All except Part (iv) of Assumption \ref{assm_u} are standard, with Part (iii) asserting that voters' utilities exhibit strict increasing differences between policies and types, hence pro-$\beta$ voters more prefer pro-$\beta$ policies to pro-$\alpha$ ones than pro-$\alpha$ voters do. Part (iv) of Assumption \ref{assm_u} says that all pro-$\beta$ voters (by symmetry, pro-$\alpha$ voters) value the policies much differently than the median voter does. Under the assumption that $A_{c}$ is a finite subset of $\Theta_c$, this condition holds true for many standard utility functions. Notable examples  including $u(a,t)=-|t-a|$, for which $\kappa \geq 2\min A_{\beta}$, and $u(a,t)=-\left(t-a\right)^2$, for which $\kappa\geq 4 \min A_{\beta}$. 

The next definition builds on Lemma \ref{lem_foc} and will prove useful:\footnote{Under Assumptions \ref{assm_symmetry} and \ref{assm_u}, we can never have $\overline{m}_t=1$, and the other condition for interior solution $\mathbb{E}_{\sigma}\left[\exp\left(-v\left(\widetilde{\bf{a}}, t\right)/\mu \right) \right] < 1$ is automatically satisfied. } 
\begin{defn}\label{defn_attention}
Under any symmetric profile $\sigma$ of candidate strategies, a voter $t<0$ is said to \emph{pay attention to politics} if $\overline{m}_t \in (0,1)$, or equivalently
\[\mathbb{E}_{\sigma}\left[\exp\left(v\left(\widetilde{\bf{a}}, t\right)/\mu \right) \right] \geq 1.\] 
\end{defn}

%said to \emph{act based on ideology} if $\overline{m}_t=0$, or equivalently
%\[\mathbb{E}_{\sigma}\left[\exp\left(v\left(\widetilde{\bf{a}}, t\right)/\mu \right) \right] < 1,\] 
%and he is 

The remainder of this section is organized as follows: Sections \ref{sec_example} presents an illustrative example; Section \ref{sec_matrix} develops a matrix representation of equilibrium outcomes; Section \ref{sec_eqm} conducts equilibrium analysis; Sections \ref{sec_evidence} and \ref{sec_discussion} discuss empirical evidence and related issues. 

\subsection{An Illustrative Example}\label{sec_example}
The next example illustrates the phenomenon of \emph{attention-driven extremism}: 

\begin{example}[label=exa:cont1]
Each candidate can be of either the centrist type $\left(t=\pm t_c\right)$ or the extreme type $\left(t=\pm t_e\right)$ with equal probability, and he can adopt either the centrist position ($a=\pm a_c$), the moderate position ($a=\pm a_m$) or the extreme position ($a=\pm a_e$). The winner and loser's utility function is $R-\delta_+|t-a|$ and $\delta_-|t-a|$, respectively, where $R \geq 0$ represents the office rent, and $\delta_+$ and $\delta_-$ the preference weights on policies. There are three groups of voters with types $t=-\tau, 0, \tau$. Their utility function is $u\left(a,t\right)=-|t-a|$. 

Consider equilibria in which candidates adopt pure, symmetric and non-degenerate strategies. The last restriction is satisfied by all equilibria of the numerical example presented below, and it will be disposed of in later sections. Let $-a_2<-a_1 \leq 0 \leq a_1<a_2$ denote the policies of the varying types of the candidates, each realized with probability $1/2$. A typical policy profile is then ${\bf{a}}_{ij}=\left(-a_i, a_j\right)$, $i, j=1,2$. 

The phenomenon of our interest is composed of two parts: 

\paragraph{Part I: equilibrium policies}  We first characterize equilibrium policies. To this end, we need to aggregate the decision probabilities across individual voters and convert the result into winning probabilities.

Consider first the median voter. By symmetry, this voter equally prefers both candidates on average, i.e., \[\overline{m}_0=\frac{1}{2},\] and an immediate corollary is that \emph{he always pays attention to politics}. Under any given policy profile, he chooses the closest candidate to the center most often, so in particular, 
\[m_0\left({\bf{a}}_{21}\right)>\frac{1}{2}. \]

Consider next pro-$\alpha$ and pro-$\beta$ voters. By Lemma \ref{lem_foc} (iv), these voters are more supportive of the candidates from their own ideological camps, i.e., \[\overline{m}_{-\tau}<\frac{1}{2}<\overline{m}_{\tau}.\] Had we known their average decision  probabilities, we could simply plug them into the logit rule and back out the decision probabilities  under each policy profile. In general, this proof strategy does not work well, because solving the average decision probability has proven to be difficult by the existing RI literature. 

To circumvent this challenge, we instead exploit symmetry, as well as the basic properties of voters' preferences and attention strategies. To illustrate, consider the problem of aggregating voters' decisions under ${\bf{a}}_{21}$:  

\begin{enumerate}[Step 1]
\item By symmetry, we can convert the problem of aggregating symmetric voters' decisions under the same policy profile to that of comparing the same voter's decisions across symmetric policy profiles: 
\[
m_{-\tau}\left({\bf{a}}_{21}\right)+m_{\tau}\left({\bf{a}}_{21}\right)
=1-m_{\tau}\left({\bf{a}}_{12} \right) +  m_{\tau}\left({\bf{a}}_{21}\right)
\]

\item Since the probability that one supports a candidate is increasing in the value thereof, we can further reduce the above problem to that of comparing a single voter's differential utilities  across symmetric policy profiles, i.e., by Lemma \ref{lem_foc},
\[
\sgn m_{\tau}\left({\bf{a}}_{21}\right) - m_{\tau}\left({\bf{a}}_{12}\right)=\sgn v\left({\bf{a}}_{21}, \tau\right)-v\left({\bf{a}}_{12}, \tau\right).
\]

\item Since voters have concave preferences for policies, it follows that
\[
v\left({\bf{a}}_{21}, \tau\right)-v\left({\bf{a}}_{12}, \tau\right) = u\left(a_1, \tau\right)+u\left(-a_1,\tau\right)-\left[u\left(a_2, \tau\right) + u\left(-a_2, \tau\right)\right]  \\
\geq  0.
\]
That is, candidate $\beta$ is more preferred among pro-$\beta$ voters when he is closer to the center rather than the other way around. 

Combining Steps 1-3, we obtain that \[m_{\tau}\left({\bf{a}}_{21}\right)+m_{-\tau}\left({\bf{a}}_{21}\right)\geq 1,\] i.e., candidate $\beta$ is as popular as, if not more so than candidate $\alpha$ is among pro-$\alpha$ and pro-$\beta$ voters. 

\item Using the median voter to break the tie, i.e., $m_0\left({\bf{a}}_{21}\right)>1/2$, we find that candidate $\beta$ wins the election for sure under ${\bf{a}}_{21}$. 
\end{enumerate}

Theorem \ref{thm_main} (i) of 
Section \ref{sec_extremism} generalizes the above described result, showing that under any policy profile of any symmetric equilibrium, the closest candidate to the center wins for sure, whereas equidistant candidates from the center split the votes evenly. While limited attention creates uncertainties for individual voters, it imposes no aggregate-level uncertainty on candidates, for whom the winner is determined the same as in  \cite{downs}.\footnote{The converse of Step 3 says that when the concavity assumption fails, candidate $\beta$ can lose the election despite being closer to the center than candidate $\alpha$ is, thus distinguishing our result from \cite{downs}. }  

\paragraph{Part II: voters' attention} Consider next the policies that capture voters' limited attention. As discussed earlier, the median voter always pays attention to politics, whereas the pro-$\alpha$ and pro-$\beta$ voters do so if and only if
\[
\sum_{i,j \in \{1,2\}}\frac{1}{4} \exp\left(v\left({\bf{a}}_{ij},-\tau\right) \right/\mu) \geq 1.
\]
Tedious but straightforward algebra as in the proof of Lemma \ref{lem_attentionset} reduces the above condition to
\begin{equation}\label{eqn_lowerbound}
 a_2-a_1 \geq \text{ a term strictly increasing in } \mu. 
\end{equation}
The message is clear: as marginal attention cost increases, retaining voters' attention becomes harder and requires that the varying types of the candidates adopt more different policies than before. 

\paragraph{Attention-driven extremism}  Figure \ref{figure_combine} summarizes our results. Under the assumption that the winning candidate has a stronger policy preference than the losing candidate does,\footnote{This assumption is inspired by \cite{calvert}'s original thesis, which asserts that ``Previous political trades or connections make it necessary to treat the policy concerns of some important supporters as a constraint on the candidate's action, who would then behave as though those concerns were his own.'' To us, this suggests that the strength of the policy preference can depend on the winning status, and that the winner can indeed face more constraints than the loser does. The same assertion is made by  \cite{banks} and \cite{callanderwilkie}, which essentially assume that $R> 0$, $\delta_+>0$ and $\delta_-=0$.} the game has two equilibria as depicted in diamonds, in which the centrist candidate adopts the centrist position, whereas the extreme candidate adopts either the moderate position or the extreme position. As demonstrated earlier, factors that affect the randomness of individual votes, such as marginal attention cost, are irrelevant in the determination of equilibrium policies. 

Meanwhile, the black line depicts the boundary at which pro-$\alpha$ and pro-$\beta$ voters are indifferent between (1) paying attention as in the shaded area, and (2) paying no attention as in the unshaded area. As marginal attention cost increases, this line moves northwest, causing an expansion of the unshaded area and a contraction of the shaded area. 

When marginal attention cost is low, the right-hand side of Condition (\ref{eqn_lowerbound}) is close to zero, hence all equilibria capture all voters' attention. As attention cost increases, a truncated subset of the equilibria capture voters' attention, and those that do so exhibit a greater degree of policy differential than before. In Section \ref{sec_extremism}, we argue why it makes sense to use visibility as an equilibrium selection device, especially when the dissemination of political information is costly and involves entities that profit from voters' eyeballs. If so, then increases in the marginal attention cost will leave us with a truncated set of policy distributions that exhibits a greater degree of extremism than ever. This phenomenon, formalized in Theorem \ref{thm_main} of Section \ref{sec_extremism}, is what we term as attention-driven extremism. 

\bigskip 

\begin{figure}[!h]
   \centering
    \includegraphics[width=8cm]{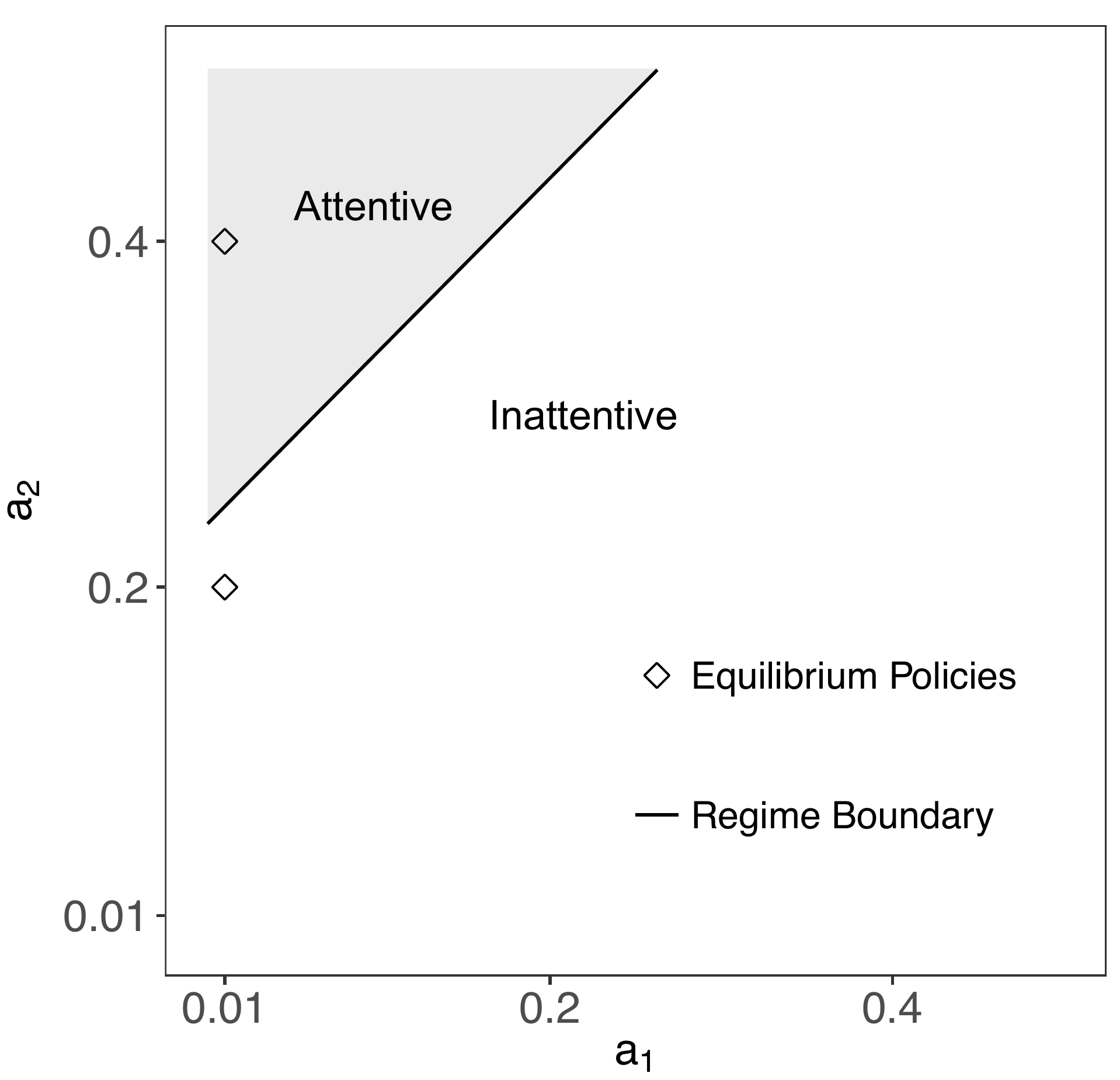}
   \caption{Equilibrium outcomes: $a_c=.01$, $a_m=.2$, $a_e=.4$, $t_c=0.3$, $t_e=0.8$, $R=8$, $\delta_+=12$ and $\delta_-=1$; regime boundary is drawn for $\tau=.001$ and $\mu=10$. }\label{figure_combine}
\end{figure}
\end{example}

\begin{rem}We will discuss the following issues after the statement of the main result: (1) supporting evidence, (2) equilibrium existence, and (3) the roles of policy preference, finite policy space and attention function in the predictions. 
\end{rem}

\subsection{Matrix Representation}\label{sec_matrix}
Let $-a_N<\cdots<-a_1 < 0 < a_1<\cdots<a_N$ denote the policies of the varying types of the candidates, where $N$ is an endogenous integer greater than one if not all types of the same candidate propose the same policy. Let $\bf{A}$, $\bm{\Sigma}$ and $\bf{W}$ be square matrices of order $N$, and denote their $ij^{th}$ entries by ${\bf{a}}_{ij}$, $\sigma_{ij}$ and $w_{ij}$, respectively. Assume throughout that (1) ${\bf{a}}_{ij}=\left(-a_i, a_j\right)$, $\sigma_{ij} > 0$, $\sigma_{ij}=\sigma_{ji}$ and $w_{ij} \in \left\{0, 1/2,1\right\}$ for all $i$ and $j$, and that (2) $\sum_{i, j} \sigma_{ij}=1$. 

 Intuitively, ${\bf{A}}$ and $\bm{\Sigma}$ compile all policy profiles and their probabilities and will therefore be referred to as the \emph{policy matrix} and the \emph{probability matrix}, respectively. Likewise, ${\bf{W}}$ collects candidate $\beta$'s winning probability under each policy profile and will thus be called the \emph{winning probability matrix}. To give a few examples, recall that the probability matrix in the leading example is $\frac{1}{4}{\bf{J}}_2$, where ${\bf{J}}_2$ denotes the $2 \times 2$ matrix of all ones. In the case where policies can be observed by voters without errors as in \cite{downs}, the winning probability matrix is $\widehat{\bf{W}}_N$ with $ij^{th}$ entry being: 
\begin{equation}\label{eqn_symmetry_wp}
\widehat{w}_{ij}=\begin{cases}
0 & \text{ if } i>j,\\
\frac{1}{2} & \text{ if } i=j,\\
1 & \text{ if } i<j.
\end{cases} 
\end{equation}
%the closest candidate to the center wins for sure, whereas equidistant candidates from the center split the votes evenly. 

A tuple $\left[{\bf{A}}, {\bm{\Sigma}}, {\bf{W}}\right]$ can be attained in a symmetric equilibrium of the election game if (1) $\left[ {\bf{A}}, {\bm{\Sigma}} \right]$ is \emph{incentive compatible for the  candidates under $\bf{W}$} (hereinafter, $\bf{W}$-IC), and if (2) $\bf{W}$ can be \emph{rationalized by optimal attention strategies under $\left[{\bf{A}}, {\bm{\Sigma}}\right]$} (hereinafter, $[{\bf{A}}, \bm{\Sigma}]$-rationalizable). Formally, 

\begin{defn}\label{defn_ic}
$[ {\bf{A}}, {\bm{\Sigma}} ]$ is \emph{$\bf{W}$-IC} if there exists a profile $\sigma$ of candidate  strategies such that 
\begin{enumerate}[(i)]
\item the probabilities of the policy profiles under $\sigma$ are given by ${\bm{\Sigma}}$, i.e., 
\begin{equation*}
\sigma({\bf{a}}_{ij})=\sigma_{ij} \text{ } \forall i, j;
\end{equation*}

\item each $\sigma_c$ maximizes candidate $c$'s expected payoff, taking the winning probability matrix ${\bf{W}}$ and the strategy $\sigma_{-c}$ of candidate $-c$ as given, i.e., 
\begin{equation*}
\sigma_c \in \argmax_{\sigma_c'} V_c\left({\bf{W}}, \sigma_c', \sigma_{-c}\right), 
\end{equation*}
where $V_c\left({\bf{W}}, \sigma\right)$ can be obtained from plugging ${\bf{W}}$ and $\sigma$ into Equation (\ref{eqn_candidate}). 
\end{enumerate}
\end{defn}

\begin{defn}\label{defn_rationalizable}
\emph{$\bf{W}$ is $[{\bf{A}}, {\bm{\Sigma}}]$-rationalizable} if  
\begin{equation*}
w_{ij}=w ({\bf{a}}_{ij}) \text{ } \forall i, j, 
\end{equation*}
where $w({\bf{a}}_{ij})$ can be obtained from plugging $\left(m_t\left({\bf{a}}_{ij}\right)\right)_{t \in \Theta}$ under $[{\bf{A}}, {\bm{\Sigma}}]$ into Equation (\ref{eqn_wp}). 
\end{defn}

The matrix representation singles out an underpinning assumption, namely candidates care about voters through the channel of winning probability. This assumption or a slight departure thereof is crucial for the upcoming analysis.

\subsection{Analysis}\label{sec_extremism}
 Fix any probability matrix $\bm{\Sigma}$ of any order $N$, as well as any level $\mu>0$ of marginal attention cost. Define the set of policy matrices that can be attained in the symmetric equilibria of the election game by
\begin{align*}
\mathcal{E}\left({\bm{\Sigma}}, \mu\right)=\left\{{\bf{A}}: \exists {\bf{W}} \text{ s.t. } \begin{matrix}  
[{\bf{A}}, {\bm{\Sigma}}] \text{ is } {\bf{W}}-\text{IC}  \\
 {\bf{W}} \text{ is } [{\bf{A}}, {\bm{\Sigma}}]-\text{rationalizable}
\end{matrix} \right\}. 
\end{align*}
Define the \emph{attention set} of any voter $t<0$ by the policy matrices that induce the voter to pay attention: 
\[\mathcal{A}_t\left(\bm{\Sigma},\mu\right)=\left\{{\bf{A}}: \mathbb{E}_{\left[{\bf{A}}, {\bm{\Sigma}}\right]} \left[\exp\left(v\left(\widetilde{\bf{a}}, t\right)/\mu\right) \right] \geq 1\right\}.\]
Taking intersection yields the set of policy matrices that captures the voter's attention in equilibrium:
\[\mathcal{EA}_t\left({\bm{\Sigma}}, \mu\right)=\left\{{\bf{A}} \in \mathcal{E}\left({\bm{\Sigma}},\mu\right): \mathbb{E}_{\left[{\bf{A}}, {\bm{\Sigma}}\right]} \left[\exp\left(v\left(\widetilde{\bf{a}}, t\right)/\mu\right) \right] \geq 1\right\}.\]
%as well as those that lead the voter to act based on ideology in equilibrium:
%\[\mathcal{E I}_t\left({\bm{\Sigma}}, \mu\right)=\left\{{\bf{A}} \in \mathcal{E}\left({\bm{\Sigma}},\mu\right): \mathbb{E}_{\left[{\bf{A}}, {\bm{\Sigma}}\right]} \left[\exp\left(v\left(\widetilde{\bf{a}}, t\right)/\mu\right) \right] < 1\right\}.\]
The next theorem gives characterizations of these sets: 

\begin{thm}\label{thm_main}
Assume Assumptions \ref{assm_symmetry} and \ref{assm_u}.  Then,
\begin{enumerate}[(i)]
\item  for any probability matrix $\bm{\Sigma}$ of any order $N$, the set $\mathcal{E}\left(\bm{\Sigma},\mu\right)$ is independent of $\mu$ and is given by 
\[\mathcal{E}\left(\bm{\Sigma}\right)=\left\{{\bf{A}}: [{\bf{A}}, {\bm{\Sigma}}] \emph{ is } \widehat{\bf{W}}_{N}-\emph{IC}\right\}; 
\]

\item let $\bm{\Sigma}$ be any probability matrix of any order $N \geq 2$ for which the set $\mathcal{E}\left({\bm{\Sigma}}\right)$ is non-empty. Then for any $t<0$ such that $v({\bf{a}}, t)>0$ for some ${\bf{a}} \in {\bf{A}}\in \mathcal{E}(\bm{\Sigma})$, the following happen as $\mu$ grows from zero to infinity: 
\begin{enumerate}
\item the set $\mathcal{EA}_t\left({\bm{\Sigma}}, \mu\right)$ shrinks; %and the set $\mathcal{EI}_t\left({\bm{\Sigma}}, \mu\right)$ expands;
\item $\min\left\{u\left(a_1,0\right)-u\left(a_N,0
\right): {\bf{A}} \in \mathcal{EA}_t\left({\bm{\Sigma}}, \mu\right)\right\}$ increases;
\item the objects in Parts (a) and (b) do not always stay constant.
\end{enumerate}
\end{enumerate}
\end{thm}

Part (i) of Theorem \ref{thm_main} shows that in every symmetric equilibrium of the election game, winner is determined the same as in \cite{downs}, meaning that under any given policy profile, the closest candidate to the center wins for sure, whereas equidistant candidates from the center split the votes evenly. Knowing the winning probability matrix, voters' characteristics such as marginal attention cost are irrelevant in the determination of equilibrium policies. While these characteristics certainly affect individual voters' decisions, they impose no aggregate-level uncertainty on candidates, for whom it is as if policies could be observed by voters without errors.

Part (ii) of Theorem \ref{thm_main} shows that as the marginal attention cost increases, a truncated subset of the equilibria captures voters' attention through enlarging the policy difference between the varying types of the candidates than before. Remarkably, this result holds true for any non-degenerate policy distribution that can arise in equilibrium, as well as any moderate voter who would occasionally prefer the opposing candidate in the case where policies could be observed without errors. 

At this moment, one may wonder why we as analysts should care about equilibrium attention levels. In order to address this question, recall the assumption shared by our baseline model and many existing models, namely candidates care about voters only through the channel of winning probability. What is thus ignored is the practical consideration of visibility, of eyeballs, which is the ``fuel of campaign advertising'' according to Roger Ailes.\footnote{Roger Ailes was the Chairman and CEO of Fox News and Fox Television Stations, from which he resigned in July 2016. He served as a media consultant for Republican presidents Richard Nixon, Ronald Reagan and George H. W. Bush, and was an adviser to the Donald Trump campaign. \url{https:/ /en.wikipedia.org/wiki/Roger_Ailes}.} The stylized model developed in Appendix \ref{sec_selection} formalizes this consideration, assuming that the dissemination of political information is costly and non-automatic, and it involves entities such as revenue-maximizing content providers that profit from voters' eyeballs.\footnote{\cite{berr} reports the ``handsome'' political ads revenues earned by media companies during the 2018 midterm election. } In the case where these entities incur no loss in equilibrium, we obtain a stronger result: increases in the marginal attention cost leave us with a truncated set of policy distributions that exhibits a greater degree of extremism than ever. The next section details a few accounts that support this view. 

The two interpretations differ by welfare implications. If we take Theorem \ref{thm_main} at face value, then increases in the marginal attention cost have no effect on candidates but make voters worse off.\footnote{By the envelope theorem, e.g., \cite{milgromsegal}, a slight increase in $\mu$ decreases voter's utility by a negative amount. } If we stick to the second interpretation, then as the marginal attention cost increases: (1) candidates can be better off if they end up being closer to their ideal positions; (2) voters can be better off as intuited in  \cite{abramowitzbook} if the rising policy differential triggers attention and enables more informed decisions. 

\subsection{Evidence}\label{sec_evidence}
\paragraph{The conformity in the 1950s} During the 1950s, a sense of uniformity pervaded the  American society. Economically, the country experienced marked growth and prosperity  hallmarked by the rapid expansions of the manufacturing and home construction sectors. Politically, the Cold War and fear of communism helped created a conservative and stable climate, tainted by the quasi-confrontations that slowly got intensified over time. Accordingly, the United States in 50's is considered both socially conservative and highly materialistic in nature. Conformity and conservatism set the social norms of the time.  
	
Among other things, Eisenhower's embrace of New Deal and the conservative Democrats' control of the congress left remarkable footprints on the country's political landscape. Fueled by the lack of perceived differences between  candidates, the public's attitude towards politics was best characterized by apathy and indifference, with unengaged and uninformed voters spanning through all age groups, ethnicities and occupations. Based on the National Study of Presidential Elections conducted by the University of Michigan, \cite{campbell} reports that in 1956, about one fourth of the population demonstrated familiarity with less than one out of the two issues presented to them. Even among the informed segment of the population, only 40 to 60 percent could articulate the candidates' positions  and answer questions such as whether a candidate was a conservative left or a liberal left.  

According to \cite{campbell}, candidates responded, at least at the regional  level, to the pressure from the other side. Over time, they imposed increasing constraints on policy choices, forgoing those with miniature differences and keeping the remainder that set the varying types of themselves far apart. Visibility seemed to be the foremost concern here. According to some interviewees, had they been more permissive and acted the same as before, they could not always cross the threshold of public awareness, which was vital for the messages to even ``reach the newspaper and dinner tables.'' Based on the rich interactions with regional candidates, \cite{campbell} concludes that ``In the electoral as a whole where the level of attention is so low, what the public is exposed to must be highly visible--even stark--if it is to have an impact on opinion.''

The above reported case confirms the prediction of Theorem \ref{thm_main}. The national political and economic environment during the Eisenhower era resulted in an increase in the marginal attention cost. The lack of perceived differences between national-level candidates created apathy and indifference among the public. To regain the needed visibility, local candidates exercised caution and deliberation, and the restricted policy choices exhibited increasing degrees of extremism over time.

\paragraph{Campaign messages} In Appendix \ref{sec_commitment}, we investigate an extension in which the winning candidate can only partially honor his policy proposal, now regarded as his promise made during the campaign. The result thereof, namely high attention cost leads to the selection of extreme and exaggerated campaign messages, is a mere understatement of reality. As Roger Ailes once described his role as campaign strategist:``... eyeball is the fuel and becomes difficult to grab over time... Everything I did, from issues to speeches to ads to debate rhetorics, was to keep the varying types of the candidates far apart...'' (\cite{popkin}; \cite{vavreck}).

\subsection{Discussions}\label{sec_discussion}
\paragraph{Equilibrium existence} By Theorem \ref{thm_main} (i), we can solve the equilibrium policies by considering an augmented finite game, in which candidates' payoffs depend solely on who is closer to the center and what their type realizations and policy positions are (there are finite of them by assumption). The existence of symmetric (mixed strategy) equilibria in this game follows from \cite{cheng}:\footnote{By Harsanyi (1973), mixed strategy equilibria can be viewed as the limit of the pure strategy equilibria of a perturbed game in which players have private information about their payoffs as in our model.  }

\begin{cor}\label{cor_existence}
Let everything be as in Theorem \ref{thm_main}. Then the election game admits symmetric equilibria. 
\end{cor}

\paragraph{Policy differential} The equilibria of Example \ref{exa:cont1} exhibit non-trivial degrees of policy differentials because the winning candidate has a stronger policy preference than the losing candidate does.  Under this assumption, candidates prefer to lose than to win with undesirable positions and hence are reluctant to move away from their bliss points.\footnote{By contrast, the assumption of finite policy space is not as crucial for producing policy differentials, and this can be seen from two angles. First, even in the baseline model, we can have situations that attain the outcomes of Example \ref{exa:cont1} in equilibrium even if the policy space is continuous. A simple (and rather trivial) example would be $R=0$, $\delta_-=0$, $u_+(\cdot, t_c)$ single-peaked at $.01$ and $u_+(\cdot, t_e)$ double-peaked at $.2$ and $.4$, where the last assumption captures the reality discusseed in \cite{calvert} that the winning candidate may have multiple interest groups to serve. 

Second, in the extension to noisy news--which, in our opinion, more accurately paints the reality but is probably too difficult to be studied upfront--the existence of symmetric equilibria follows from \cite{cheng} in the case where (1) the policy space is compact and convex, and (2) the news distribution is log-supermodular and is continuous in policy positions. Even so, moving towards the center is still shown to increase the winning probability, suggesting, once again, that the policy differential as depicted in Figure \ref{figure_noisy} should be the result of policy preferences rather than the topological properties of the policy space. See Example \ref{exa:cont2} for further discussions.}

To compare and contrast, consider the situations investigated by (1) \cite{downs}, in which $R>0$ and $\delta_+=\delta_-=0$, and (2) \cite{calvert}, in which $R\geq 0$ and $\delta_+=\delta_->0$. In both cases, the unique equilibrium outcome is $\left(a_1,a_2\right)=\left(a_c, a_c\right)$, in contrast to what we have in Example \ref{exa:cont1}. The reality may lie somewhere in between or be even more different,\footnote{For example, \cite{kartik} essentially assumes that the candidate is either non-strategic with $\delta_+=\infty$ and $t$ being drawn from a continuous distribution, or is strategic with $R>0$ and $\delta_+=\delta_-=0$. } which we make no judgment about. The mere goal here is to clarify the role of policy preference in producing policy differentials. 

\paragraph{Attention in the middle} Theorem \ref{thm_main} (i) builds on the following lemma, saying that the middle is always attentive to non-degenerate policy distributions (at a level that could be arbitrarily low): 

\begin{lem}\label{lem_interior}
Assume Assumption \ref{assm_symmetry} and fix any $\mu>0$. Then in any symmetric equilibrium in which the policy distribution is non-degenerate, there exists a neighborhood around the median voter that pays attention to politics. 
\end{lem}

\begin{proof}
The proof exploits symmetry and the fact that $\exp(-x)+\exp(x)>2$ for all $x \neq 0$. See Appendix \ref{sec_proof_eqm} for further details. 
\end{proof}

Lemma \ref{lem_interior} adds to the on-going debate over whether the middle attends to politics or not---an issue which we will continue to narrate in Appendix \ref{sec_voter}. Its proof is one of the few places that exploit the functional form of mutual information, leaving us to wonder about the validity of Theorem \ref{thm_main} in more general environments. 

\paragraph{Alternative attention function} The answer to the above question is an affirmative one.  Recall the equilibrium construction in Example \ref{exa:cont1}, which has four steps. Steps 1-3 exploit symmetry and basic monotonicity properties and should thus carry over to more general settings. To break the tie, we invoke Lemma \ref{lem_interior} in Step 4, but the proof also works if (1) voters have strictly concave preferences for policies, and (2) some symmetric pair of voters is actively attending to politics. Thus, the only difference it makes from assuming alternative attention functions is the following: after the marginal attention cost exceeds a threshold, all voters, including the median one, could stop paying attention, and candidates best respond to the winning probability $\frac{1}{2}{\bf{J}}_N$ instead. Theorem \ref{thm_main} remains valid beforehand.
%\footnote{In the earlier draft of this paper, we formalized the above statement in specific environments, and the material is available upon request.}

\paragraph{Attention set} The proof of Theorem \ref{thm_main} (ii) combines Theorem \ref{thm_main} (i) with the following characterization of the attention set: 

\begin{lem}\label{lem_attentionset}
Let everything be as in Theorem \ref{thm_main}. As $\mu$ grows from zero to infinity, the set $\mathcal{A}_t\left({\bm{\Sigma}}, \mu\right)$ satisfies the description of $\mathcal{EA}_t\left(\bm{\Sigma}, \mu\right)$ in Theorem \ref{thm_main}. 
\end{lem}

\begin{proof}
See Appendix \ref{sec_proof_eqm}.
\end{proof}

Two things are noteworthy. First, the lemma itself, albeit sounding straightforward, is  not so easy to prove. Second, one should not confuse Theorem \ref{thm_main} (ii) with this mere lemma, because had equilibrium policies depended on the marginal attention cost, too, we could arrive at the opposite conclusion: as marginal attention cost increases, we could end up with more rather than fewer equilibria that grab voters' attention!

\section{Noisy News}\label{sec_noisy}
\subsection{Setup}\label{sec_noisy_setup}
In this section, let news be a profile $\bm{\omega}=\left(\omega_{\alpha}, \omega_{\beta}\right)$ of reports, each drawn independently from a finite set $\Omega_c \subset \Theta_c$ according to a probability mass function $f_c\left(\cdot \mid a_c\right)$. The joint probability distribution $f=f_{\alpha}\times f_{\beta}$ has support  $\Omega=\Omega_{\alpha}\times \Omega_{\beta}$, and it is called the \emph{news technology}. 

 In this modified setting, attention strategy is a mapping $m_t: \Omega \rightarrow [0,1]$, where each $m_t\left(\bm{\omega}\right)$ specifies the probability that voter $t$ chooses candidate $\beta$ under news profile $\bm{\omega}$. For any given pair $x=\left(f,\sigma\right)$ of news technology and candidate strategies, let 
\begin{equation*}
\nu_{x}\left(\omega, t\right)=\mathbb{E}_{x}\left[v\left(\widetilde{\bf{a}}, t\right) \mid \bm{\omega}\right]
\end{equation*}
be the voter's expected differential utility from choosing candidate $\beta$ over candidate $\alpha$ under news profile $\bm{\omega}$, and let 
\begin{equation*}
V_t\left(m_t, x\right)=\mathbb{E}_{x}\left[m_t\left(\widetilde{\bm{\omega}}\right) \nu_{x}\left(\widetilde{\bm{\omega}}, t\right)\right].
\end{equation*}
be the expected differential utility under $\left(m_t,x\right)$. The net payoff is then
\begin{equation*}
V_t\left(m_t,x\right) - \mu \cdot I\left(m_t, x\right), 
\end{equation*}
where $I\left(m_t, x\right)$ is the mutual information between news and the voting  decision. Under any profile $(m,x)$ of strategies and news technology, the expected payoff of candidate $c$ is equal to
\begin{equation*}
V_c\left(m, x\right)= \mathbb{E}_{m,x}\left[w_c\left(\widetilde{\bm{\omega}}\right) u_+\left(\widetilde{a}_c, \widetilde{t}_c\right) + \left(1-w_c\left(\widetilde{\bm{\omega}}\right) \right) u_-\left(\widetilde{a}_{-c}, \widetilde{t}_c\right) \right].
\end{equation*}
In the above expression, $w_c\left(\bm{\omega}\right)$ represents candidate $c$'s winning probability under news profile $\bm{\omega}$ and can be obtained from replacing $\bf{a}$ with $\bm{\omega}$ in Equation (\ref{eqn_wp}). 

\subsection{Equilibrium Analysis}\label{sec_noisy_eqm}
Under news technology $f$, a strategy profile $\left(m^*, \sigma^*\right)$ constitutes a \emph{Bayes Nash equilibrium} of the election game if each player maximizes expected utility, taking $f$ and the other players' strategies as given. In what follows, we assume that the environment, which now includes the news technology, is symmetric around the median voter: 

\begin{assm}\label{assm_noisy_symmetry}
$f_{\alpha}\left(-\omega \mid -a\right)=f_{\beta}\left(\omega \mid a\right)$ for all $a \in A_{\beta}$ and $\omega \in \Omega_{\beta}$.  
\end{assm}

Again, we focus on the symmetric equilibria of the election game:
 
\begin{defn}
A strategy profile $\left(m, \sigma\right)$ is \emph{symmetric around the median voter} if 
\begin{enumerate}[(i)]
\item $\sigma_{\beta} \left(a\mid t\right)=\sigma_{\alpha}\left(-a \mid -t\right)$ for all $a \in A_{\beta}$ and $t \in \supp\left(P_{\beta}\right)$;
\item  $m_t\left(-\omega,\omega'\right)=1-m_{-t}\left(-\omega', \omega\right)$ for all $\left(-\omega,\omega'\right)\in \Omega$ and $t \in \Theta$. 
\end{enumerate}
\end{defn} 

The analysis assumes that the news technology is log-supermodular  (\cite{milgrom}), meaning that we are more likely to observe extreme rather than centrist news reports as the underlying policy becomes more extreme: 

\begin{assm}\label{assm_logsupermodular}
The following holds true for all $c$, $a, a' \in A_c$ and $\omega, \omega' \in \Omega_c$ such that $a<a'$ and $\omega<\omega'$: 
\[\frac{f_c\left(\omega' \mid a\right)}{f_c\left(\omega \mid a\right)}<\frac{f_c\left(\omega' \mid a'\right)}{f_c\left(\omega \mid a'\right)}.\]
\end{assm}

The next definition extends Definition \ref{defn_attention} to encompass noisy news: 
\begin{defn}
Under any symmetric profile $x=(f,\sigma)$ of news technology and candidate strategies, a voter $t<0$ is said to \emph{pay attention to politics} if 
\[\mathbb{E}_{x}\left[\exp\left(\nu_x\left(\widetilde{\bm{\omega}}, t\right)/\mu \right) \right] \geq 1.\] 
\end{defn}

% \emph{act based on ideology} if 
%\[\mathbb{E}_{x}\left[\exp\left(\nu_x\left(\widetilde{\bm{\omega}}, t\right)/\mu \right) \right] < 1,\] 
%and he is said to 

\subsubsection{Matrix Representation}\label{sec_noisy_matrix}
Let $-\omega_K<\cdots<-\omega_1<0< \omega_1<\cdots<\omega_K$ denote the realizations of news signals, where $K$ is an exogenous integer greater than one. Write $\bm{\omega}_{mn}=\left(-\omega_m, \omega_n\right)$ for $m, n=1,\cdots, K$. Let ${\bf{A}}$ and $\bm{\Sigma}$ be as above, and let $\bf{W}$ be the $K \times K$ matrix whose $mn^{th}$ entry $w_{mn} \in \left\{0,1/2,1\right\}$ represents candidate $\beta$'s winning probability under news profile $\bm{\omega}_{mn}$. 

 Under news technology $f$, a tuple $\langle {\bf{A}}, {\bm{\Sigma}}, {\bf{W}}\rangle$ of matrices can be attained in a symmetric equilibrium of the election game if (1) $[{\bf{A}}, {\bm{\Sigma}}]$ is \emph{incentive compatible for the candidates under} $\langle f, {\bf{W}} \rangle$, and if (2) $\bf{W}$ \emph{can be rationalized by optimal attention strategies under} $\langle f, {\bf{A}}, {\bf{\Sigma}} \rangle $ (hereinafter, $\langle f, {\bf{W}} \rangle$-IC and $\langle f, {\bf{A}}, {\bm{\Sigma}}\rangle$-rationalizable, respectively). Formally, 

\begin{defn}
$[ {\bf{A}}, {\bm{\Sigma}} ]$ is \emph{$\langle f,  {\bf{W}}\rangle $-IC} if in Definition \ref{defn_ic}, the winning probability matrix is replaced with an $N \times N$ matrix whose $ij^{th}$ entry is 
\begin{equation*}
\sum_{m,n=1}^K w_{mn} f\left(\bm{\omega}_{mn} \mid {\bf{a}}_{ij}\right) \sigma_{ij}. 
\end{equation*}
\end{defn}

\begin{defn}
$\bf{W}$ is \emph{$\langle f, {\bf{A}}, {\bm{\Sigma}} \rangle $-rationalizable}  if 
\begin{equation*}
w_{mn}=w\left(\bm{\omega}_{mn}\right) \text{ } \forall m, n, 
\end{equation*}
where $w\left(\bm{\omega}_{mn}\right)$ can be obtained from plugging $\left(m_t\left({\bm{\omega}}_{mn}\right)\right)_{t \in \Theta}$ under $\langle f, {\bf{A}}, {\bm{\Sigma}}\rangle $ into Equation (\ref{eqn_wp}). 
\end{defn}

\subsubsection{Analysis}
\paragraph{Equilibrium policies} For any pair $\left(\bm{\Sigma},f\right)$, define the set of policy matrices that can be attained in the symmetric equilibria of the election game by 
\begin{align*}
\mathcal{E}\left({\bm{\Sigma}}, f\right)=\left\{{\bf{A}}: \exists {\bf{W}} \text{ s.t. }   \begin{matrix}  
[{\bf{A}}, {\bm{\Sigma}}] \text{ is } \langle f, {\bf{W}} \rangle-\text{IC}  \\
{\bf{W}} \text{ is } \langle f, {\bf{A}}, {\bm{\Sigma}} \rangle-\text{rationalizable}
\end{matrix}\right\}
\end{align*}
The next theorem gives a full characterization of this set:
\begin{thm}\label{thm_noisy_policy}
Assume Assumptions \ref{assm_symmetry}-\ref{assm_logsupermodular}. Then for any probability matrix $\bm{\Sigma}$ of any order $N \geq 2$,\footnote{The case of $N=1$ is less interesting: $\mathcal{E}\left({\bm{\Sigma}}, f\right)=\left\{{\bf{A}}: 
[{\bf{A}}, {\bm{\Sigma}}] \text{ is } \langle f, \frac{1}{2}{\bf{J}}_{K}\rangle-\text{IC} \right\}$.} 
 
\[
\mathcal{E}\left({\bm{\Sigma}}, f\right)=\left\{{\bf{A}}: 
[{\bf{A}}, {\bm{\Sigma}}] \emph{ is } \langle f, \widehat{\bf{W}}_{K}\rangle-\emph{IC} \right\}. \]
\end{thm}

\begin{proof}
See Appendix \ref{sec_proof_noisy}. 
\end{proof}

Theorem \ref{thm_noisy_policy} shows that in any symmetric equilibrium in which candidates adopt non-degenerate strategies, \emph{winner is determined as if news were fully revealing}, meaning that under any news profile, the candidate earning the most centrist news report wins for sure, whereas those earning equidistant news reports from the center split the votes evenly. 

The implication of Theorem \ref{thm_noisy_policy} is twofold. First, the key insight of Theorem \ref{thm_main} (i)--namely voters' characteristics such as marginal attention cost are irrelevant in the determination of equilibrium policies--remains valid. As before, equilibrium policies can be solved by considering an augmented game between candidates only, and the existence of symmetric equilibrium follows from \cite{cheng}.\footnote{The same is true when the news distribution is continuous in policy positions and the policy space is compact and convex, yet we do not state this as a main result in order to keep consistency with the baseline model. }

Second, news technology enters the determination of equilibrium policies, and the effect is more delicate than what we have seen so far (more details to come). 

%a slight change in the policy position has only marginal effects on the winning probability, and the existence of symmetric mixed strategy equilibrium in the case of compact and convex policy space follows from \cite{cheng}.

\paragraph{Attention set}
 For any given pair $\left(\bm{\Sigma}, f\right)$, define the \emph{attention set} of an arbitrary voter $t<0$ by %the policy matrices that induce the voter to pay attention to politics under $\left(\bm{\Sigma}, f\right)$: 
\[\mathcal{A}_t\left({\bm{\Sigma}}, f\right)=\left\{{\bf{A}}: \mathbb{E}_{\langle f, {\bf{A}}, {\bm{\Sigma}} \rangle}\left[\exp\left(\mu^{-1}\nu_{\langle f, {\bf{A}}, {\bm{\Sigma}} \rangle}(\bm{\omega},t)\right)\right] \geq 1\right\}.\]
Taking intersection with $\mathcal{E}\left({\bm{\Sigma}},f\right)$ yields the set of policy matrices that draws the voter's attention to politics in equilibrium: 
\[\mathcal{EA}_t\left({\bm{\Sigma}}, f\right)=\mathcal{E}\left({\bm{\Sigma}}, f\right) \cap \mathcal{A}_t\left({\bm{\Sigma}}, f\right),\]
%
%as well as those that lead the voter to act based on ideology in equilibrium: 
% \[\mathcal{\mathcal{EI}}_t\left({\bm{\Sigma}}, f\right)=\mathcal{E}\left({\bm{\Sigma}}, f\right) \cap \mathcal{A}^c_t\left({\bm{\Sigma}}, f\right).\]
In Appendix \ref{sec_noisy_mu}, we verify that the above defined sets vary the same with the marginal attention cost as their equivalents in Theorem \ref{thm_main}. Below we examine how they depend on the informativeness of the news technology. 

As in \cite{handbooktheory}, we adopt Blackwell (1953)'s notion of informativeness: 
 
\begin{defn}
$f$ is more \emph{Blackwell-informative} than $f'$ (\emph{$f'$ is a garble of $f$}, $f \succeq f'$) if there exists a Markov kenel $\rho$ such that for all $\bf{a}$ and $\bm{\omega'}$, 
\[
f'\left(\bm{\omega}' \mid {\bf{a}}\right)=\sum_{\bm{\omega}\in \Omega}f\left(\bm{\omega}\mid {\bf{a}}\right)\rho\left(\bm{\omega}' \mid \bm{\omega}\right). 
\]
\end{defn}

Since \cite{blackwell}, it is well known that garbling destroys information and makes decision makers worse-off in all decision problems. Recently, economists, journalists and political scientists have voiced concerns for the rise of partisan media and fake news (\cite{levendusky};  \cite{handbooktheory}; \cite{handbookempirical}; \cite{gentzkowfakenews}). We seek to understand how these concerns, modeled as garblings of the news technology, can affect equilibrium outcomes through the channel of limited voter attention.

We first examine the effect of garbling on the attention set: 

\begin{thm}\label{thm_noisy_attention}
Assume Assumptions \ref{assm_symmetry} and \ref{assm_u}. Fix any probability matrix $\bm{\Sigma}$ of any order $N \geq 2$, as well as any $f \succeq f'$ that satisfy Assumptions \ref{assm_noisy_symmetry} and \ref{assm_logsupermodular}. Then for any $t<0$ such that the sets $\mathcal{A}_t\left({\bm{\Sigma}}, f\right)$ and $\mathcal{A}_t\left({\bm{\Sigma}},f'\right)$ are different and non-empty, 
\begin{enumerate}[(i)]
\item $\mathcal{A}_t\left({\bm{\Sigma}},f'\right)\subsetneq \mathcal{A}_t\left({\bm{\Sigma}},f\right)$;
\item $\displaystyle \min_{{\bf{A}} \in \mathcal{A}_t\left({\bm{\Sigma}},f'\right)}\nu_{\langle f'', {\bf{A}}, {\bm{\Sigma}}\rangle}\left(\bm{\omega}_{K1},0\right)>\displaystyle \min_{{\bf{A}} \in \mathcal{A}_t\left({\bm{\Sigma}},f\right)}\nu_{\langle f'', {\bf{A}}, {\bm{\Sigma}}\rangle}\left(\bm{\omega}_{K1},0\right)$ for $f''=f,f'$. 
\end{enumerate}
\end{thm}

\begin{proof}
See Appendix \ref{sec_proof_noisy}. 
\end{proof}

Theorem \ref{thm_noisy_attention} conveys an intuitive message: as the news technology becomes less Blackwell-informative, retaining voters' attention becomes harder and requires that the policy differential between the varying types of the candidates be greater on average. To develop intuition, notice that 
\[\nu_{\langle f, {\bf{A}}, {\bm{\Sigma}}\rangle}\left(\bm{\omega}_{K1},0\right)=\underbrace{\mathbb{E}_{\langle f,{\bf{A}}, {\bm{\Sigma}}\rangle}\left[u\left(\widetilde{a}_{\beta},0\right) \mid \omega_{\beta}=\omega_1\right]}_\text{(1)}-\underbrace{\mathbb{E}_{\langle f,{\bf{A}}, {\bm{\Sigma}}\rangle}\left[u\left(\widetilde{a}_{\beta},0\right) \mid \omega_{\beta}=\omega_K\right]}_\text{(2)},\]
where (1) and (2) is the median voter's expected valuation of candidate $\beta$'s policy upon hearing the most centrist news report and the most extreme news report, respectively. Expanding these terms, we obtain
\[(1)=\frac{\sum_{i=1}^N f_{\beta}\left(\omega_1 \mid a_i\right)u(a_i,0)\sigma_i}{\sum_{i=1}^N f_{\beta}\left(\omega_1 \mid a_i\right)\sigma_i},\]
and 
\[(2)=\frac{\sum_{i=1}^N f_{\beta}\left(\omega_K \mid a_i\right)u(a_i,0)\sigma_i}{\sum_{i=1}^N f_{\beta}\left(\omega_K \mid a_i\right)\sigma_i},\]
where $\sigma_i=\sum_{j} \sigma_{ji}$ is the marginal probability that candidate $\beta$'s policy is $a_i$. By log-supermodularity, we know that (1) weighs centrist policies heavily and (2) extreme policies heavily. Thus, an enlarging gap between (1) and (2) means that the centrist and extreme policies have become more different on average, holding probabilities and the news technology fixed.

%The latter quantity is measured by the median voter's utility differential from choosing candidate $\beta$ over candidate $\alpha$ upon hearing the most centrist news report about the former and the most extreme news report about the latter. 

\paragraph{Media-driven extremism} The overall effect of garbling is illustrated by the next example:

\begin{example}[label=exa:cont2]
In Example \ref{exa:cont1}, let the policy spaces $A_c$'s can be  arbitrarily rich, and suppose the news signals can take either the centrist value $\left(\omega=\pm  \omega_1\right)$ or the extreme value $\left(\omega=\pm \omega_2\right)$, $0<\omega_1<\omega_2<1$. 
The news technology is $f_{\xi}=f_{\alpha,\xi} \times f_{\beta, \xi}$, where $f_{\beta,\xi}\left(\omega_2 \mid a\right)=a+\xi(1-a)$ is the probability that we hear the extreme news report when the policy is $a$. The parameter $\xi \in (0,1)$ captures the degree of \emph{slanting}, as well as the Blackwell-informativeness of the news technology: as $\xi$ increases, we are more likely to hear extreme news reports, ceteris paribus, and the quality of news deteriorates. 

Consider equilibria in which candidates adopt pure symmetric strategies with policy realizations $-a_2 <-a_1 < 0 < a_1 < a_2$ (indeed, all equilibria take this form and hence are strict). Figure \ref{figure_noisy} depicts the results for three levels of $\xi$'s. In particular, the diamonds represent the $\left(a_1, a_2\right)$'s that can arise in equilibrium, and the lines represent the boundaries at which voter $-.001$ is indifferent between paying attention or not. 

%\bigskip 

\begin{figure}[!h]
    \centering
     \includegraphics[width=8cm]{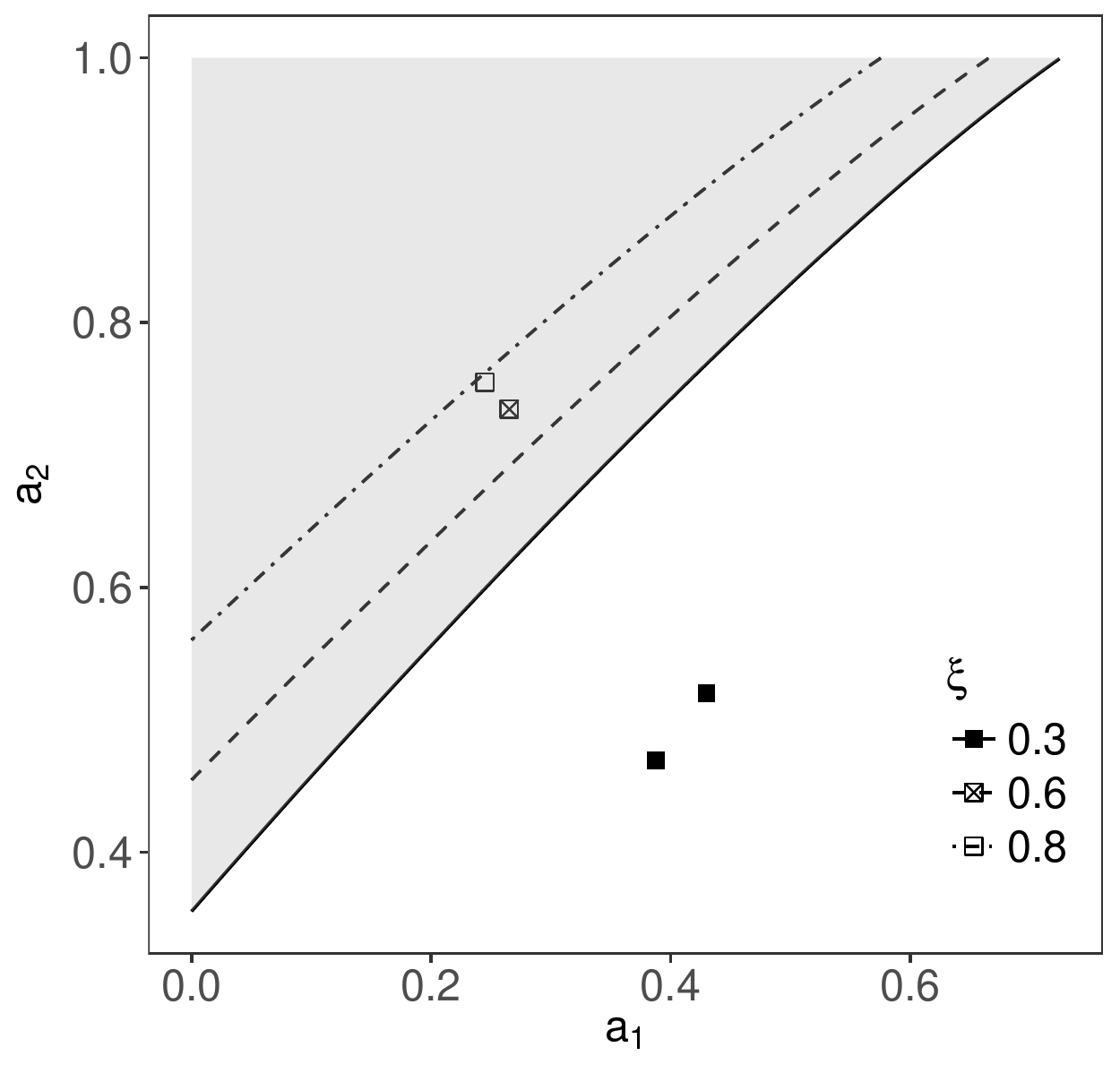}
    \caption{Equilibrium outcomes when $t_c=1/4$, $t_e=3/4$, $R=8$, $\delta_{+} = 3$ and $\delta_{-}=1$: diamonds represent the $\left(a_1,a_2\right)$'s that can arise in equilibrium when $A_{\beta}$ consists of 50 evenly spaced points in $[0,1]$; the shaded areas above the lines represent the attention sets when $\tau=.001$ and  $\mu=1$.  }\label{figure_noisy}
\end{figure}

As $\xi$ increases, the attention set depicted in the shaded area shrinks and moves northwest, and the remaining points exhibit a greater degree of policy differential than before. Equilibrium policies converge to candidates' bliss points $(1/4, 3/4)$, though the trajectory can depend subtly on their utility functions and the news technology. As in Section \ref{sec_eqm}, policy differential arises as a result of candidates' policy preferences: under the assumption of log-supermodular news, moving towards the center increases the winning probability by Theorem \ref{thm_noisy_policy}, suggesting, once again, that the winner must have a stronger policy preference than the loser does in order to prevent the convergence to the center from happening.\footnote{This feature distinguishes our noisy news from the aggregate shock to voters' preferences in probabilistic voting models, as the latter prevents the candidates as in \cite{wittman} and \cite{calvert} from moving towards the center.} 

Since garbling affects both the attention set and policies, the overall effect is more subtle than what we have seen so far. Two observations are immediate. First, any equilibrium that grabs voters' attention after garbling must exhibit a greater of policy differential than those failing to do so prior to garbling---a phenomenon we term as \emph{media-driven extremism}. 

Second, in order to back out the degree of garbling from real-world data, we need to identify shocks that affect only the attention set but not the policies, or vice versa.  By Theorem \ref{thm_main} and its corollary in Appendix \ref{sec_noisy_mu}, the attention cost shifters as  discussed in Section \ref{sec_cost} seem to serve this purpose well. 
\end{example}

\section{Concluding Remarks}\label{sec_conclusion}
As discussed in Section \ref{sec_discussion}, a main prediction of our analysis--which is robust to the attention function being used--is that equilibrium policies can be solved by considering an augmented incomplete information game between candidates only. To make further progress--e.g., prove existence of pure strategy equilibrium, give conditions for equilibrium uniqueness or multiplicity, provide detailed characterizations--we need to make additional assumptions about candidates and the news technology, but not about voters or their attention function. For now, we choose not to pursue this route, because the exercise is tangential to our focus. 

For the same reason, we well acknowledge the possibility of having multiple equilibria, and believe that it indeed makes the interpretation of attention- and media-driven extremism easier rather than harder. In case the reader wishes to conduct equilibrium refinement, he or she can be assured that this exercise--which replaces the set of equilibrium policies with its subset--leaves our main messages untouched. 

Caution should be exercised when extrapolating our results to the study of elite polarization. First, such exercise commands a holistic characterization of the policy distribution and requires the same additional assumptions as discussed above. Second, as noted by \cite{barbermccarthy}, there is no necessary logical connection between polarization and the reduction in the dimensionality of political conflict as documented in \cite{poole}---a key element of which is that intra-party division loses its predictive power of roll-call vote decisions during most of the twentieth century, only to rebound to the 90 percent level in recent years. 

The lesson is twofold. First, a richer model that captures party as a team, career concerns, etc. is needed for better predicting the outcomes of congressional voting (see, e.g., \cite{polbornsnyder} for researches along this line). Second, when testing the predictive power of our model, much effort should be devoted to the identification of shocks that affect only the attention cost or the news technology but not the countervailing factors as laid out in  \cite{barbermccarthy}.  

Finally, it seems plausible to apply our result to the study of industrial organization topics, such as product differentiation against rationally inattentive consumers. We hope someone, maybe ourselves, will pursue this research agenda in the future.

\appendix

\section{Rational Inattentive Voters: Part II}\label{sec_voter}
%This section continues the investigation of rationally inattentive voters. 

\subsection{Selective Exposure, Confirmatory Bias and Bit Occasional Surprise}
\begin{figure}[ht]
\centering
\includegraphics[scale=.55]{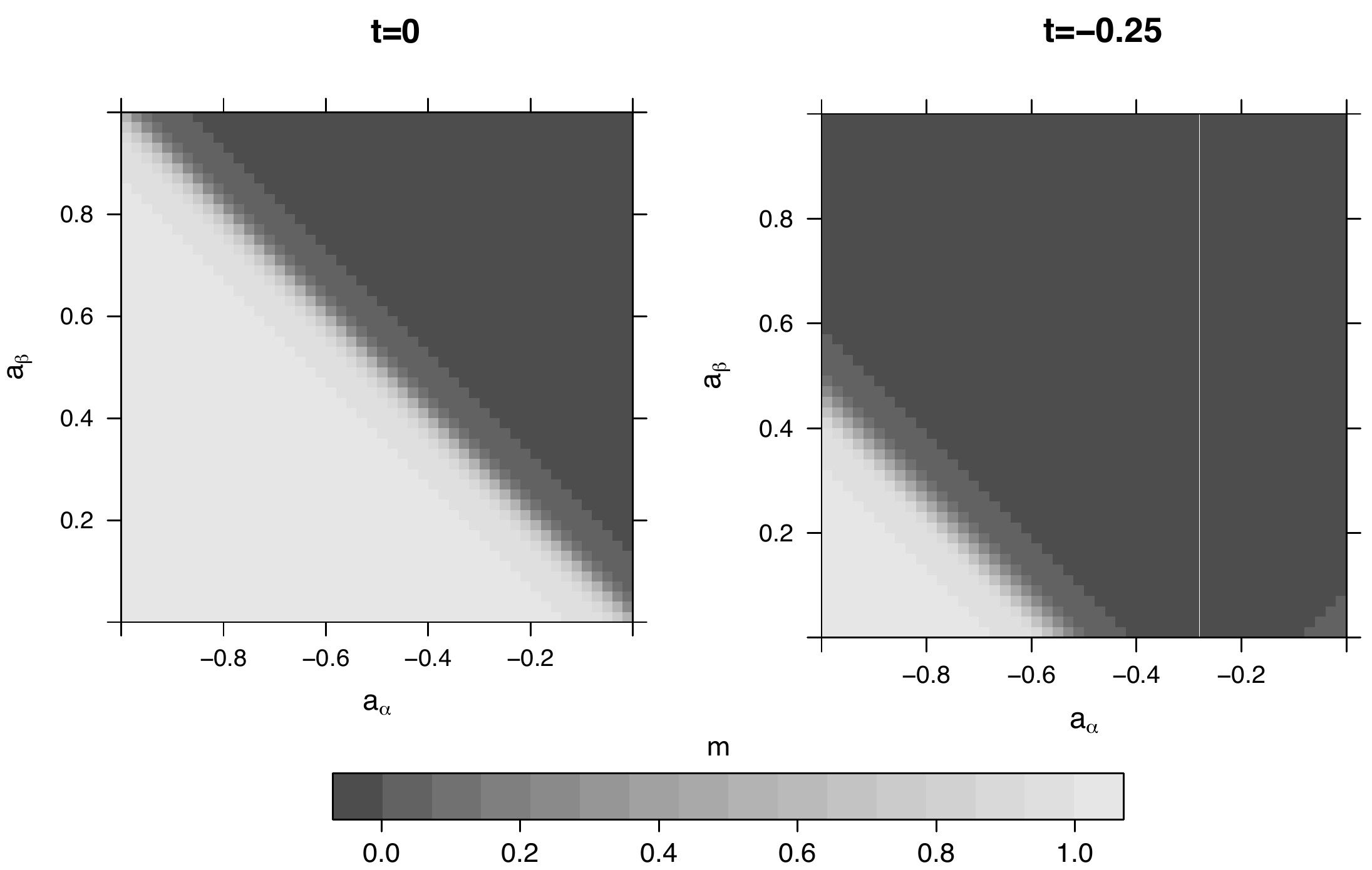}
\caption{Plot optimal attention strategies against ${\bf{a}}$ for $t=0$ and $t=-.25$: $a_c$ is uniformly distributed on $\Theta_c$, $c=\alpha, \beta$, $u(a,t)=-|t-a|$ and $\mu=.02$. }\label{figure_interpersonal_m}
\end{figure}

Figure \ref{figure_interpersonal_m} depicts the optimal attention strategies of two voters: $t=0$ and $t=-.25$. Compared to the median voter $t=0$, the pro-$\alpha$ voter $t=-.25$ more prefers candidate $\alpha$ than candidate $\beta$ ceteris paribus. Even in light of an arbitrarily small attention cost, such predisposition manifests itself through attention allocation, leading the pro-$\alpha$ voter to focus on whether candidate $\beta$ is much closer to the center than candidate $\alpha$ is, but nothing else.\footnote{An antecedent of this result appears in \cite{calvert}, showing that biased voters may benefit most from consulting biased experts when choosing between multiple information sources. } This shouldn't come as a surprise, since voters learn most from occasional surprises that challenge their predispositions when the default is to pay no attention and act based on the ideology. 

Outcome-wise, it is as if the pro-$\alpha$ voter selectively exposes himself to some pro-$\alpha$ media outlet or personalized news aggregator that gathers information from the sources on his behalf.\footnote{We emphasize the role of media outlet as information aggregator rather than information sources, and further pursue this agenda in our companion paper \cite{fragmentation}. } With high probability measured by the area of the dark area, the recommendation is to choose candidate $\alpha$, thus exhibiting confirmatory biases (Lemma \ref{lem_foc} (iv) formalizes this observation).\footnote{Confirmatory bias refers to the tendency to search for, interpret, favor and recall information in a way that confirms one's preexisting beliefs or hypotheses. The question of whether the consumption of political news and contents exhibits confirmatory bias has attracted much attention of economists and political scientists. We refer the reader to \cite{gentzkow2}, \cite{iyengar} and the references therein for thorough reviews of the literature.}  In the case where the recommendation is in favor of $\beta$, the evidence must be overwhelming in that candidate $\beta$ is known to be much closer to the center than candidate $\alpha$ is, hence even the pro-$\alpha$ voter should act accordingly.\footnote{See \cite{gentzkowgep} and the recent coverage of Fox News for examples of such recommendations. \cite{gentzkowgep} provides justifications based on competition and reputation concerns. Here and in \cite{fragmentation}, they are part of the demand of rationally inattentive voters. }

%This section studies voters' equilibrium behaviors. The results compiled in Tables \ref{table1} and \ref{table2} reveal two patterns. First, among voters in the same ideological camp, the attention level, measured by the mutual information between the policy profile and the voting decision, can be non-monotone in the decision-maker's predisposition, even if we keep the marginal attention cost the same for everyone. Second, for the same voter, the average decision probability (and hence the pointwise decision  probabilities) can be non-monotone in the marginal attention cost. 

\subsection{Non-Monotonicity}\label{sec_nonmonotone}
An intriguing consequence of rational, flexible attention allocation is that voter behaviors can vary non-monotonically with their personal characteristics such as political predisposition and marginal attention cost.\footnote{Such flexibility pertains to, but is not unique to, RI attention allocation, and similar conclusions can be drawn for alternative posterior-separable attention functions. The material is available upon request. }  Below we conduct  analysis of attention levels and decision probabilities, holding the equilibrium policy distribution fixed.

\paragraph{Who pays more attention?}
Consider first the level of paid attention, measured by the mutual information between policies  and the voting decision. Table \ref{table1} compiles the results of three voters: $t=0$, $t=.05$ and $t=.2$. The median voter, by definition, holds a neutral position. While he focuses mainly on whether $a_{\alpha}+a_{\beta}\geq 0$ or not, he still distinguishes between the policy profiles within each category when attention is not too costly. By contrast, voter $t=-.2$ has a strong predisposition to support candidate $\alpha$, whereas voter $t=-.05$ focuses sharply on whether ${\bf{a}}=(-.4, .-.01)$ or not due to his slight preference for candidate $\alpha$. In the end, voter $t=-.05$ pays most attention, followed by the median voter and then voter $t=-.2$. 

%\bigskip 

\begin{table}
\begin{center}
  \begin{tabular}{ | c | c | c | c | c | c | }
    \hline
    $t$ & $I_t$ &  $m_t\left(-.01,.01\right)$ 
	& $m_t\left(-.01,.4\right)$ 
	& $m_t\left(-.4,.01\right)$ 
	& $m_t\left(-.4,.4\right)$ \\ \hline \hline
    -.2 & 0 & 0 & 0 & 0 & 0\\ \hline
    -.05 & .315 & .296&.006 & .930 & .148\\ \hline
    0 & .312 & .500 & .012 & .987 & .500\\
    \hline
  \end{tabular}
  \caption{Attention levels of various voters: policies equal $\pm .01$ and $\pm .4$ with probability $1/2$, $u(a,t)=-|t-a|$ and $\mu=.09$.}\label{table1}
   \end{center}
\end{table}

\paragraph{Effect of marginal attention cost} Consider next the average propensity to choose candidate $\beta$. Table \ref{table2} compiles the result of voter $t=-.05$, which first increases with the marginal attention cost and then decreases after the latter reaches a threshold. Interestingly, what drives this pattern seems to be the flexibility that pertains to, but is not unique to, RI attention allocation.

As $\mu$ increases, the voter garbles his decisions at the varying policy profiles. During the first phase, the degree of garbling is similar across all policy profiles, meaning that the pointwise decision probability is inflated if it is initially below the average and is deflated otherwise. Since there are three policy profiles in the former category and one in the latter category, the average decision probability mechanically increases. During the second phase, most garbling occurs at the policy profile $\left(-.4, .01\right)$ where the voter would have chosen candidate $\beta$ under perfect information but does so less often now. As a result, the average decision probability decreases. 

%\bigskip 
%under the three policy profiles where its initial values are below the average, and it is deflated under the single policy profile where its initial value is above the average. This explains why the average decision probability first increases. 

\begin{table}
\begin{center}
  \begin{tabular}{ | c | c | c | c | c | c | }
    \hline
    $\mu$ & $\overline{m}_t$ &  $m_t\left(-.01,.01\right)$ 
	& $m_t\left(-.01,.4\right)$ 
	& $m_t\left(-.4,.01\right)$ 
	& $m_t\left(-.4,.4\right)$ \\ \hline \hline
    .01 & .261 & .046 & .000 & 1 & .000\\ \hline
    .10 & .344 & .300 & .009 &  .905 & .162\\ \hline
    .20 & .283 & .263 & .048  & .627 & .148 \\
    \hline
  \end{tabular}
  \caption{Probabilities of choosing candidate $\beta$ at various levels of $\mu$'s: policies equal $\pm .01$ and $\pm .4$ with probability $1/2$, $u(a,t)=-|t-a|$ and $t=-.05$. }\label{table2}
   \end{center}
\end{table}

\paragraph{The role of flexibility} To facilitate comparison, consider a more stylized setting in which voters can pay a cost $c(I)$ and fully observe the policy profile with probability $p(I)$, where $p$ and $c$ are increasing smooth functions defined on $\mathbb{R}_+$.  For any pro-$\alpha$ voter, the problem becomes to choose the level (but not the allocation) of attention that maximizes 
\[p(I)\cdot \mathbb{E}_{\sigma}\left[v\left(\widetilde{\bf{a}},t\right)\right]-\mu \cdot c\left(I\right).\]Under the assumption that $u$ has strict increasing differences in $(a,t)$ and hence $v\left({\bf{a}}, t\right)$ is  strictly increasing in $t$, it is easy to show that $I_t$ is increasing in $t$ and $\overline{m}_t$ is decreasing in $\mu$ among pro-$\alpha$ voters.

 \paragraph{Lessons} By now, there is a broad scholarly consensus that extreme voters are activists, that they are more attentive to and engaged in politics (\cite{abramowitzbook}). Evidence pertaining to the middle is mixed (\cite{barbermccarthy}), with some asserting that the middle is apathetic, unengaged and uninformed (\cite{abramowitzbook}), and others insisting that the antipathy is exaggerated and is a measurement issue at large (\cite{broockman}). Our theory sheds light on this debate. On the one hand, we can induce high levels of attention among extreme voters by setting their marginal attention costs low, or by invoking the ``refocusing argument'' as provided in  the explanation of Table 1. On the other hand, Lemma \ref{lem_interior} shows that in environments like ours, the middle is always attentive, albeit at a level that could be arbitrarily low. 
 
The question of whether mass polarization is on the rise has triggered heated debates among the public. When it comes to measuring polarization, economists prefer to look at actual votes rather than self-reported political views, because the former is ``Closer to the kind of `revealed preference' measure economists prefer'' (\cite{gentzkow}). Our addition is twofold. First, if measured behaviors are the result of voters paying limited attention to their political surroundings, then changes in the information environment per se can already cause variations thereof, even if the policy distribution is held fixed. Thus, one should not equate measurements with intrinsic preferences and should instead tease them apart using the methods proposed and reviewed in \cite{caplindean}. Second, in light of the subtle results presented in Table 2, one should not be too surprised that the evidence on mass polarization as documented in \cite{fiorina} and \cite{gentzkow} is at best mixed.

%\begin{figure}[ht]
%\centering
%\includegraphics[scale=.55]{mbar_mu.pdf}
%\caption{Plot $\overline{m}_t$ against $\mu$ for $t=-.05$: policies equal $0$ and $\pm .4$ with equal probability, $u(t,a)=-|t-a|$. }\label{figure_interpersonal_m}
%\end{figure}
%
%\bigskip

%\begin{figure}[ht]
%\centering
%\includegraphics[scale=0.65]{I_t.pdf}
%\caption{Plot mutual information against $t$: policies equal $0$ and $\pm 0.4$ with equal probability, $u(t,a)=-|t-a|$ and $\mu=.09$. }\label{figure_interpersonal_i}
%\end{figure}

\section{Omitted Proofs}\label{sec_proof}

\subsection{Proofs of Section \ref{sec_cost}}\label{sec_proof_cost}
\noindent Proof of Lemma \ref{lem_foc} (iv)
\begin{proof}
The problem of voter $t$ boils down to choosing the likelihood ratio $\Lambda$ that maximizes:
\[\mathbb{E}_{\sigma}\left[v\left(\widetilde{\bf{a}},t\right) \frac{\Lambda\exp\left(\frac{v\left(\widetilde{\bf{a}},t\right)}{\mu}\right)}{\Lambda\exp\left(\frac{v\left(\widetilde{\bf{a}},t\right)}{\mu}\right)+ 1}\right]-\mu \cdot I\left(\Lambda, \sigma\right),\]
where $I\left(\Lambda, \sigma \right)$ stands for the mutual information between the policies and the voting decision. Under Assumption \ref{assm_u} (iv), the above objective function has strict increasing differences in $\left(\Lambda, t\right)$, and the result follows from \cite{milgromshannon}. 
\end{proof}

\subsection{Proofs of Section \ref{sec_eqm}}\label{sec_proof_eqm}
 Define the function $\gamma: \mathbb{R} \rightarrow \mathbb{R}$ by 
\begin{equation}\label{eqn_gamma}
\gamma(x)=\exp(x)+\exp(-x), 
\end{equation}
and notice the following properties: (1) $\gamma \geq 2$ and the equality holds at $x=0$, (2) $\gamma'>0$ on $\mathbb{R}_+$, (3) $\gamma''>0$, and (4) for any $b>1$, the equation $\gamma(x)=2b$ has a unique positive root
	\begin{equation}
	\gamma^{-1}(2b)=\log\left(b+ \sqrt{b^2-1}\right)	
	\end{equation}

In what follows, write 
\begin{equation}
\Delta_{ij}(t)=\mu^{-1}v\left({\bf{a}}_{ij}, t\right), 
\end{equation}
and drop the notation of $t$ in the case of $t=0$. 

The next lemma is useful for the upcoming analysis:  
\newtheorem{lemma}{Lemma} \setcounter{lemma}{-1}
\begin{lemma}\label{lem_symmetry}
Under Assumptions \ref{assm_symmetry} and \ref{assm_u}(i), the following hold true for any policy matrix $\bf{A}$ of any order $N$: 
\begin{enumerate}[(i)]
\item $\Delta_{ij}=-\Delta_{ji}$ $\forall i, j=1,\cdots, N$;
\item $\Delta_{ij}>0$ $\forall i>j$;
\item $\Delta_{N1}=\displaystyle \max_{i,j}\Delta_{ij}$.
\end{enumerate}
\end{lemma}

\begin{proof}
In the evaluation of $v\left({\bf{a}}_{ij},0\right)=u\left(a_i,0\right)-u\left(a_j,0\right)$, using the assumption that $u\left(\cdot, 0\right)$ is strictly increasing on $[-1,0]$, strictly decreasing on $[0,1]$ and symmetric around $a=0$ gives the desired result.
\end{proof}

\bigskip 

\noindent Proof of Lemma \ref{lem_interior}
\begin{proof}
Fix any non-degenerate strategy profile for the candidates. 
Let $\left[{\bf{A}}, {\bm{\Sigma}}\right]$ be the corresponding matrix representation, where ${\bf{A}}$ and $\bm{\Sigma}$ are square matrices of order $N \geq 2$. By Lemma \ref{lem_symmetry}, 
\begin{align*}
\mathbb{E}_{[{\bf{A}}, {\bm{\Sigma}}]} \left[\exp\left(\widetilde{\Delta}_{ij}\right) \right]
&=\sum_{i=1}^N \sigma_{ii} \exp\left(\Delta_{ii}\right) + \sum_{i=1}^N \sum_{j=1}^{i-1} \sigma_{ij} \gamma\left(\Delta_{ij}\right)\\
&>\sum_{i=1}^N \sigma_{ii} \cdot 1 + \sum_{i=1}^N \sum_{j=1}^{i-1} \sigma_{ij} \cdot 2=1, 
\end{align*}
and $\mathbb{E}_{[{\bf{A}}, {\bm{\Sigma}}]} \left[\exp\left(-\widetilde{\Delta}_{ij}\right)\right]>1$, where both inequalities are strict because $N\geq 2$. Then by continuity, the above inequalities are strict in a neighborhood $I$ of $t=0$, and the result follows from Lemma \ref{lem_foc}. 
%\begin{align*}
%&\mathbb{E}_{[{\bf{A}}, {\bm{\Sigma}}]} \left[\exp\left(\mu_0^{-1}v\left(\widetilde{\bf{a}}, 0\right)\right) \right]\\
%=&\sum_{i=1}^N \sigma_{ii} \exp\left(\mu_0^{-1}v\left({\bf{a}}_{ii},0\right)\right) + \sum_{i=1}^N \sum_{j=1}^{i-1} \sigma_{ij} \gamma\left(\mu_0^{-1}v\left({\bf{a}}_{ij},0\right)\right)\\
%>&\sum_{i=1}^N \sigma_{ii} \cdot 1 + \sum_{i=1}^N \sum_{j=1}^{i-1} \sigma_{ij} \cdot 2=1,
%\end{align*}
\end{proof}

\noindent Proof of Lemma \ref{lem_attentionset}
\begin{proof}
Let $t$ and $\bm{\Sigma}$ be as in Theorem \ref{thm_main}. Below we demonstrate that as  $\mu$ grows from zero to infinity, 
\begin{enumerate}[(i)]
\item the set $\mathcal{A}_t\left(\bm{\Sigma}, \mu\right)$ shrinks;
\item the term $\min\left\{u(a_1,0)-u(a_N,0): {\bf{A}} \in \mathcal{A}_t(\bm{\Sigma}, \mu)\right\}$ increases;
\item the objects in Parts (i) and (ii) do not always stay constant. 
\end{enumerate}

\bigskip 

\noindent Part (i): Let $\bf{A}$ be any policy matrix of order $N$. By optimality, the mutual information between the policy profile induced by $[{\bf{A}}, {\bf{\Sigma}}]$ and the voter's optimal decision is decreasing in $\mu$, and using this fact in the assessment of the set $\mathcal{A}_t\left(\bm{\Sigma}, \mu\right)$ gives the desired result. 

\bigskip

\noindent Part (ii): Let $\bf{A}$ be as above. By Assumption \ref{assm_u} (iii) and (iv),  
\begin{align*} 
\mathbb{E}_{[{\bf{A}}, {\bm{\Sigma}}]}\left[\exp\left(\widetilde{\Delta}_{ij}(t)\right)\right] 
\leq \exp\left(\kappa t/\mu\right) \mathbb{E}_{[{\bf{A}},{\bm{\Sigma}}]}\left[\exp\left(\widetilde{\Delta}_{ij}\right)\right], 
\end{align*}
where the last term of the above inequality can be bounded above as follows:
\begin{align*}
\mathbb{E}_{[{\bf{A}},{\bm{\Sigma}}]}\left[\exp\left(\widetilde{\Delta}_{ij}\right)\right]
=&\sum_{i=1}^N \sigma_{ii} \exp\left(\Delta_{ii}\right)+ \sum_{i=1}^N \sum_{j=1}^{i-1} \sigma_{ij} \gamma\left(\Delta_{ij}\right)\\
\leq & \sum_{i=1}^N \sigma_{ii} \cdot 1+\frac{1}{2}\left(1-\sum_{i=1}^N \sigma_{ii}\right)\gamma\left(\Delta_{N1}\right).    
\end{align*}
Thus, a necessary condition for $\mathbb{E}_{[{\bf{A}}, {\bm{\Sigma}}]}\left[\exp\left(\widetilde{\Delta}_{ij}\left(t\right)\right)\right]\geq 1$ to hold true is  
 \[
 	\exp(\kappa t /\mu)\left[\sum_{i=1}^N \sigma_{ii} +\frac{1}{2}\left(1-\sum_{i=1}^N \sigma_{ii}\right) \gamma\left(\Delta_{N1}\right)\right] \geq 1,\] 
 or equivalently
	\[\gamma\left(\Delta_{N1}\right) \geq 2b, \]
	where 
	\[b \triangleq \frac{\exp(\kappa \lvert t \rvert/ \mu ) -\sum_{i=1}^N \sigma_{ii}}{1-\sum_{i=1}^N \sigma_{ii}}>1.\] %\[x\triangleq \mu^{-1}v\left({\bf{a}}_{N1},0\right)=\mu^{-1}[u(a_{1},0)-u(a_{N},0)],\] and	
 Since $\gamma'>0$ on $\mathbb{R}_{+}$ and the equation $\gamma(x)=2b$ has a unique positive root
	\[\gamma^{-1}(2b)=\log\left(b+ \sqrt{b^2-1}\right),\]
     the above necessary condition is equivalent to 
	\[
	u(a_{1},0)-u(a_{N},0) \geq \mu \gamma^{-1}(2b)\triangleq \delta(\mu).
	\]

	It is easy to verify that $\delta(\mu)$ is strictly positive for all $\mu>0$. Below we demonstrate that it is strictly increasing in $\mu$:
	
	\begin{description}
	
	\item[Step 1] Show that $\frac{d\delta(\mu)}{d \mu}>0$ when $\mu>\frac{\kappa}{2\log2}$. For ease, write $y= \exp(\kappa\lvert t \rvert/ \mu)$, $z=b+\sqrt{b^2-1}$  and $\Sigma=\sum_{i=1}^N \sigma_{ii}$. Notice that $y, z>1$ and $\Sigma<1$. 
	
	Differentiating $\delta(\mu)$ with respect to $\mu$, we obtain
		\[
		\frac{d\delta(\mu)}{d\mu}=\log z - \frac{y\log y}{(1-\Sigma)\sqrt{b^2-1}}.
		\]
		Since 
		\[\log z  =	\log [1+(z-1)]  \geq \frac{2(z-1)}{2+(z-1)}=\frac{2(z-1)}{z+1},\]
        and 
        \[\log y = \log[1+(y-1)]  \leq \frac{(y-1)}{\sqrt{1+(y-1)}}=\frac{y-1}{\sqrt{y}}, \]
		a sufficient condition for $\frac{d\delta(\mu)}{d\mu}>0$ to hold true is  
		\[
		\frac{2(z-1)}{z+1} - \frac{y \cdot \frac{y-1}{\sqrt{y}}}{(1-\Sigma)\sqrt{b^2-1}}>0.	
		\]
		Plugging $z=b+\sqrt{b^2-1}$ and $b=\frac{y -\Sigma}{1-\Sigma}$ into the above inequality and rearranging, we obtain 
		\begin{multline*}
			2\left(y-1+ \sqrt{(y-1)(y+1-2\Sigma)}\right)\sqrt{y+1-2\Sigma}\\ -\left(y+1-2\Sigma)+\sqrt{(y-1)(y+1-2\Sigma}\right)\sqrt{y(y-1)}>0.
		\end{multline*}
		This is equivalent to
		\[
		\sqrt{(y-1)(y+1-2\Sigma)} \left(\sqrt{y-1}+\sqrt{y+1-2\Sigma}\right)\left(2-\sqrt{y} \right)>0, 
		\]
		which holds true if and only if $2-\sqrt{y}>0$, or equivalently   $\exp(\kappa \lvert t \rvert /2 \mu) <2$. A sufficient condition is $\mu > \frac{\kappa}{2\log2}$.
		
		\item[Step 2] Show that $\frac{d^2 \delta(\mu)}{d\mu^2}<0$. Since 
		\[
			\frac{d^2 \delta(\mu)}{d\mu^2} = 
			- \frac{\log y}{\sqrt{y+1-2\Sigma}}\left(1-  \frac{y(y-\Sigma)}{y+1-2\Sigma} \right)\frac{d y}{d\mu}, 
		\]
	a sufficient condition for $\frac{d^2 \delta(\mu)}{d\mu^2}<0$ to hold true is
		\[
		1-  \frac{y(y-\Sigma)}{y+1-2\Sigma}<0, 
		\]
		or equivalently $y>2-\frac{1}{\Sigma}$. This is indeed the case since $y>1$ and $\Sigma<1$.
		\end{description}
		
	\bigskip 
	
\noindent Part (iii): When $\mu \approx 0$, $\mathcal{A}_t\left(\bm{\Sigma}, \mu\right) \neq \emptyset$ from the assumption that $v\left({\bf{a}},t\right)>0$ for some ${\bf{a}} \in {\bf{A}} \in \mathcal{E}\left({\bm{\Sigma}}\right)$. When $\mu$ is large, %the following holds true for any policy matrix ${\bf{A}}$ of order $N$: 
\begin{align*}
\mathbb{E}_{\left[{\bf{A}}, {\bm{\Sigma}}\right]}\left[\exp\left(\widetilde{\Delta}_{ij}(t)\right)\right] 
\underbrace{\approx}_\text{(1)} \mathbb{E}_{\left[{\bf{A}}, {\bm{\Sigma}}\right]}\left[1+\widetilde{\Delta}_{ij}(t)\right]
\underbrace{<}_\text{(2)}\mathbb{E}_{\left[{\bf{A}}, {\bm{\Sigma}}\right]}\left[1+\widetilde{\Delta}_{ij}\right]
= 1,
\end{align*}
where (1) uses the fact that $\exp(z)\approx 1+z$ when $z \approx 0$, and (2) Assumption \ref{assm_u} (iii).  Thus $\mathcal{A}_t\left(\bm{\Sigma}, \mu\right)=\emptyset$, and this completes the proof. 
\end{proof}

\noindent Proof of Theorem \ref{thm_main}
\begin{proof}
Part (i): By symmetry, the following holds true for all $i$ and $j$:  
\begin{align*}
\int m_t^*\left({\bf{a}}_{ij}\right) dF(t)=&\int_{t<0} 1-m_{-t}^*\left({\bf{a}}_{ji}\right) dF(t) + \int_{t>0} m_t^*\left({\bf{a}}_{ij}\right) dF(t)\\
=&\int_{t>0} m_t^*\left({\bf{a}}_{ij}\right)-m_t^*\left({\bf{a}}_{ji}\right) dF(t) + \frac{1}{2}. 
\end{align*}
When $N=1$, the result is immediate because the first term of the above line is equal to zero. 

Suppose $N \geq 2$. Below we argue that for all $i>j$, (1) $m_t^*\left({\bf{a}}_{ij}\right) \geq m_t^*\left({\bf{a}}_{ji}\right)$ for all $t$, and (2) the inequality is strict in a neighborhood of $t=0$: 

\begin{description}
\item[Step 1] Since the function $u(\cdot,t)$ is concave and hence the function $u(\cdot,t)+u(-\cdot,t)$ is decreasing in its argument, the following holds true for all $t$: 
\[
v\left({\bf{a}}_{ij},t\right)-v\left({\bf{a}}_{ji},t\right)
=u\left(a_j, t\right)+u\left(-a_j,t\right)-\left[u\left(a_i, t\right)+u\left(-a_i,t\right)\right] \geq 0. 
\]
The result then follows from Lemma \ref{lem_foc}, which shows that $m_t^*({\bf{a}})$ is increasing in $v({\bf{a}}, t)$.
 
\item[Step 2] By continuity, there exists a neighborhood $I'$ of $t=0$ in which $v({\bf{a}}_{ij}, t)>v({\bf{a}}_{ji}, t)$. Meanwhile, combining Lemmas \ref{lem_foc} and \ref{lem_interior} shows that $m_t^*({\bf{a}})$ is strictly increasing in $v({\bf{a}}, t)$ in the neighborhood $I$ of $t=0$. Intersecting these neighborhoods gives the desired result, namely $m_t^*\left({\bf{a}}_{ij}\right)>m_t^*\left({\bf{a}}_{ji}\right)$ for all $t \in I \cap I'$.
\end{description}

\bigskip

\noindent Part (ii): Combining the result of Part (i) and Lemma \ref{lem_attentionset} gives the desired result. 
\end{proof}

\subsection{Proofs of Section \ref{sec_noisy}}\label{sec_proof_noisy}

\subsubsection{Useful Lemmas}
\begin{lem}\label{lem_noisy_foc}
For any given $x=\left(f, \sigma\right)$, the optimal attention strategy of any voter $t$ uniquely exists. Let $m_t: \Omega \rightarrow [0,1]$ be as such, and let \[\overline{m}_t=\mathbb{E}_{x}\left[m_t\left(\widetilde{\bm{\omega}}\right)\right]\] be the average probability that voter $t$ chooses $\beta$ under $x$. 
Then, 
\begin{enumerate}[(i)]
\item if $\mathbb{E}_{x}\left[\exp\left(\nu_x\left(\widetilde{\bm{\omega}}, t\right)/\mu\right) \right] < 1$, then $\overline{m}_t=0$;
\item if $\mathbb{E}_{x}\left[\exp\left(-\nu_x\left(\widetilde{\bm{\omega}}, t\right)/\mu\right) \right] <1$, then $\overline{m}_t=1$;
\item otherwise $\overline{m}_t \in \left(0,1\right)$ and for any $ \bm{\omega}$: 
\[m_t\left({\bm{\omega}}\right)=\frac{\Lambda_t\exp\left(\frac{\nu_x\left({\bm{\omega}},t\right)}{\mu}\right)}{\Lambda_t\exp\left(\frac{\nu_x\left({\bm{\omega}},t\right)}{\mu}\right)+ 1}, \]
where 
\[\Lambda_t=\frac{\overline{m}_t}{1-\overline{m}_t}\]
is the likelihood that voter $t$ chooses $\beta$ over $\alpha$ under $x$. 
\end{enumerate}
\end{lem}

\begin{proof}
The proof is analogous to that of Lemma \ref{lem_foc} and is thus omitted.
\end{proof}

The next lemma is adapted from \cite{milgrom}: 

\begin{lem}\label{lem_milgrom}
Assume Assumption \ref{assm_logsupermodular}. Take any $\sigma$ that is symmetric and non-degenerate, and write $x=(f,\sigma)$. Then $\mathbb{E}_{x}\left[h\left(\widetilde{a}_{\beta}\right) \mid \omega_{\beta}=\omega\right]$ is increasing (resp. strictly increasing) in $\omega$ if the function $h$ is increasing (resp. strictly increasing) in its argument.
\end{lem}

\begin{proof}
See \cite{milgrom}.
\end{proof}

The next lemma generalizes Lemma \ref{lem_symmetry} to encompass noisy news: 

\begin{lem}\label{lem_noisy_symmetry}
Assume Assumptions \ref{assm_symmetry}-\ref{assm_logsupermodular}. Take any $\sigma$ that is symmetric and non-degenerate, and write $x=(f,\sigma)$. Then, 
\begin{enumerate}[(i)]
\item $\nu_x\left(\bm{\omega}_{mn}, t\right) \geq \nu_x\left(\bm{\omega}_{nm}, t\right)$  $\forall t$ and $m>n$; 
\item for $t=0$, 
  \begin{enumerate}[(a)]
\item $\nu_x\left(\bm{\omega}_{mn},0\right)=-\nu_x\left(\bm{\omega}_{nm},0\right)$ $\forall m$, $n$;
\item $\nu_x\left(\bm{\omega}_{mn},0\right)>\nu_x\left(\bm{\omega}_{nm},0\right)$ $\forall m>n$;
\item $\nu_{x}\left(\bm{\omega}_{K1},0\right)=\max_{\bm{\omega}}\nu_x\left(\bm{\omega},0\right)$. 
\end{enumerate}
\end{enumerate}
\end{lem}

\begin{proof}
Part (i): By symmetry, 
\[\nu_x\left(\bm{\omega}_{mn},t\right)=\mathbb{E}_x\left[u\left(\widetilde{a}_{\beta},t\right)\mid \omega_{\beta}=\omega_n\right]-\mathbb{E}_x\left[u\left(-\widetilde{a}_{\beta},t\right)\mid \omega_{\beta}=\omega_m\right],\]
and hence
\begin{align*}
\nu_x\left(\bm{\omega}_{mn},t\right) - \nu_x\left(\bm{\omega}_{nm},t\right)&= \mathbb{E}_{x}\left[u\left(\widetilde{a}_{\beta}, t\right)+u\left(-\widetilde{a}_{\beta},t\right)\mid \omega_{\beta}=\omega_n\right]\\
&-\mathbb{E}_{x}\left[u\left(\widetilde{a}_{\beta}, t\right)+u\left(-\widetilde{a}_{\beta},t\right)\mid \omega_{\beta}=\omega_m\right],
\end{align*} 
so the problem boils down to showing that the right-hand side of the above equality is positive. To this end, notice that the term $\mathbb{E}_{x}\left[u\left(\widetilde{a}_{\beta}, t\right)+u\left(-\widetilde{a}_{\beta},t\right)\mid \omega_{\beta}=\omega\right]$ is decreasing in $\omega$ by Assumption \ref{assm_u} (ii) and Lemma \ref{lem_milgrom}.

\bigskip

\noindent Part (ii): In the above derivation, using the fact that $u(a,0)=u(-a,0)$ yields 
\begin{align*}
\nu_x\left(\bm{\omega}_{mn},0\right)&=\mathbb{E}_x\left[u\left(\widetilde{a}_{\beta},0\right)\mid \omega_{\beta}=\omega_n\right]-\mathbb{E}_x\left[u\left(-\widetilde{a}_{\beta},0\right)\mid \omega_{\beta}=\omega_m\right]\\
&=\mathbb{E}_x\left[u\left(\widetilde{a}_{\beta},0\right)\mid \omega_{\beta}=\omega_n\right]-\mathbb{E}_x\left[u\left(\widetilde{a}_{\beta},0\right)\mid \omega_{\beta}=\omega_m\right], 
\end{align*}
where the term $\mathbb{E}_x\left[u\left(\widetilde{a}_{\beta},0\right)\mid \omega_{\beta}=\omega\right]$ is strictly decreasing in $\omega$ by Lemma \ref{lem_milgrom} and Assumption \ref{assm_u} (i). The remainder the proof is straightforward and is thus omitted.
\end{proof}

The next lemma shows that garbling adds white noises to voters' differential utilities from choosing one candidate over another: 

\begin{lem}\label{lem_mps}
Fix any $t$, $\sigma$ and $f \succeq f'$, and write $x=(f,\sigma)$ and $x'=(f',\sigma)$. Then, 
\begin{enumerate}[(i)]
\item for any $\bm{\omega}' \in \Omega$, there exist probability weights $\left\{\pi\left(\bm{\omega}', \bm{\omega}\right)\right\}_{\bm{\omega}\in \Omega}$ such that 
\[\nu_{x'}(\bm{\omega}',t)=\sum_{\bm{\omega}}\pi\left(\bm{\omega}',\bm{\omega}\right)\nu_x\left(\bm{\omega},t\right).\]
\item $\mathbb{E}_{x}\left[\nu_x\left(\widetilde{\bm{\omega}},t\right)\right]=\mathbb{E}_{x'}\left[\nu_x\left(\widetilde{\bm{\omega}}',t\right)\right]$.
\end{enumerate}
\end{lem}

\begin{proof}
Let $\rho$ be the Markov kernel. For all $\bm{\omega}, \bm{\omega}' \in \Omega$, define 
\begin{equation}
\pi\left(\bm{\omega}' ,\bm{\omega}\right)=\frac{\mathbb{P}_x\left(\bm{\omega}\right)\rho\left(\bm{\omega}' \mid \bm{\omega}\right)}{\sum_{\widetilde{\bm{\omega}}}\mathbb{P}_x\left(\widetilde{\bm{\omega}}\right)\rho\left(\bm{\omega}' \mid \widetilde{\bm{\omega}}\right)}=\frac{\mathbb{P}_x\left(\bm{\omega}\right)\rho\left(\bm{\omega}' \mid \bm{\omega}\right)}{\mathbb{P}_{x'}\left(\bm{\omega}'\right)}, 
\end{equation}
and notice that $\sum_{\bm{\omega}} \pi\left(\bm{\omega}', \bm{\omega}\right)=1$ for all $\bm{\omega}'$. 

By definition, 
\begin{align*}
\mathbb{P}_{x'}\left(\bm{\omega}'\right)=\sum_{\bm{\omega}, {\bf{a}} \in \supp(\sigma)} \rho\left(\bm{\omega}' \mid \bm{\omega}\right)f\left(\bm{\omega} \mid {\bf{a}}\right)\sigma({\bf{a}})=\sum_{\bm{\omega}} \rho\left(\bm{\omega}'\mid \bm{\omega}\right) \mathbb{P}_{x}\left(\bm{\omega}\right), 
\end{align*}
and 
\begin{align*}
\nu_{x'}\left(\bm{\omega}', t\right)= & \mathbb{E}_{x'}\left[v\left(\widetilde{\bf{a}},t\right) \mid \bm{\omega}'\right]\\
= & \frac{\displaystyle \sum_{{\bf{a}} \in \supp(\sigma)} f'\left(\bm{\omega}'\mid {\bf{a}}\right) \sigma\left({\bf{a}}\right) v\left({\bf{a}},t\right) }{\displaystyle \sum_{{\bf{a}} \in \supp(\sigma)} f'\left(\bm{\omega}'\mid {\bf{a}}\right) \sigma\left({\bf{a}}\right) }\\
= & \frac{\displaystyle \sum_{{\bf{a}} \in \supp(\sigma)}\sum_{\bm{\omega}} \rho\left(\bm{\omega}'\mid \bm{\omega}\right)f\left(\bm{\omega}\mid {\bf{a}}\right)\sigma\left({\bf{a}}\right) v\left({\bf{a}},t\right) }{\displaystyle \sum_{{\bf{a}} \in \supp(\sigma)} \sum_{\bm{\omega}}\rho\left(\bm{\omega}'\mid \bm{\omega}\right) f\left(\bm{\omega}\mid {\bf{a}}\right) \sigma\left({\bf{a}}\right) }\\
=&\frac{\displaystyle \sum_{\bm{\omega}} \rho\left(\bm{\omega}' \mid \bm{\omega}\right)\mathbb{P}_x\left(\bm{\omega}\right)\nu_x\left(\bm{\omega},t\right) }{\displaystyle\sum_{\bm{\omega}} \rho\left(\bm{\omega}'\mid \bm{\omega}\right) \mathbb{P}_{x}\left(\bm{\omega}\right)}\\
= & \sum_{\bm{\omega}} \pi \left(\bm{\omega}', \bm{\omega}\right)\nu_x\left(\bm{\omega},t\right). 
\end{align*}
Thus, 
\begin{multline*}
\mathbb{E}_{x'}\left[\nu_{x'}\left(\widetilde{\bm{\omega}}',t\right)\right]
=\sum_{\bm{\omega}'}\left(\sum_{\bm{\omega}} \pi\left(\bm{\omega}',\bm{\omega}\right)\nu_x\left(\bm{\omega},t\right) \right) \cdot \mathbb{P}_{x'}\left(\bm{\omega}'\right)\\
=\sum_{\bm{\omega}} \nu_x\left(\bm{\omega},t\right) \mathbb{P}_x\left(\bm{\omega}\right) \cdot \sum_{\bm{\omega}'}\rho\left(\bm{\omega}'\mid\bm{\omega}\right)
=\mathbb{E}_x\left[\nu_x\left(\widetilde{\bm{\omega}},t\right)\right],
\end{multline*}
and this completes the proof.
\end{proof}

The next lemma shows that garbling reduces the median voter's differential utility from choosing $\beta$ over $\alpha$ upon hearing the most centrist report of the former and the most extreme report of the latter: 

\begin{lem}\label{lem_extreme}
Assume Assumptions \ref{assm_symmetry} and \ref{assm_u}. Fix any $\sigma$ and any $f \succeq f'$ that satisfy Assumptions \ref{assm_noisy_symmetry} and \ref{assm_logsupermodular}, and write $x=(f,\sigma)$ and $x'=(f',\sigma).$ Then $\nu_{x}\left(\bm{\omega}_{K1}, 0\right) \geq \nu_{x'}\left(\bm{\omega}_{K1}, 0\right).$
\end{lem}

\begin{proof}
For all $\bm{\omega}'$,  
\begin{align*}
\nu_{x'}\left(\bm{\omega}',0\right)=& \frac{\displaystyle \sum_{\bm{\omega} }\sum_{{\bf{a}}\in \supp\left(\sigma\right)}  v\left({\bf{a}}, 0\right) \rho\left(\bm{\omega}' \mid \bm{\omega}\right)f\left(\bm{\omega} \mid {\bf{a}}\right)\sigma({\bf{a}})}{\displaystyle \sum_{\bm{\omega} } \sum_{{\bf{a}}\in \supp(\sigma)} \rho\left(\bm{\omega}' \mid \bm{\omega}\right) f\left(\bm{\omega}\mid {\bf{a}}\right) \sigma\left({\bf{a}}\right)}\\
\leq & \max_{\bm{\omega}}\frac{\displaystyle \sum_{{\bf{a}}\in \supp\left(\sigma\right)}  v\left({\bf{a}}, 0\right) f\left(\bm{\omega} \mid {\bf{a}}\right)\sigma({\bf{a}})}{\displaystyle \sum_{{\bf{a}}\in \supp\left(\sigma\right)} f\left(\bm{\omega}\mid {\bf{a}}\right) \sigma\left({\bf{a}}\right)}\\
=& \max_{\bm{\omega}}\nu_x\left(\bm{\omega}, 0\right)\\
=&\nu_x\left(\bm{\omega}_{K1}, 0\right). 
\end{align*}
where the last equality uses Lemma \ref{lem_noisy_symmetry} (ii (c)). Maximizing the left-hand side of the above inequality, we obtain \[\nu_{x'}(\bm{\omega}_{K1}, 0)=\max_{\bm{\omega}'}\nu_{x'}\left(\bm{\omega}',t\right) \leq \nu_x\left(\bm{\omega}_{K1}, 0\right),\] and this completes the proof.
\end{proof}

\subsubsection{Proofs of Main Results}
\noindent Proof of Theorem \ref{thm_noisy_policy} 
\begin{proof}
In the proof of Theorem \ref{thm_main} (i), replacing $v\left({\bf{a}}, t\right)$'s with 
$\nu_x\left(\bm{\omega},t\right)$'s and invoking Lemmas  \ref{lem_noisy_foc} and \ref{lem_noisy_symmetry} gives the desired result.  

\end{proof}

\noindent Proof of Theorem \ref{thm_noisy_attention}
\begin{proof}
Fix any tuple $[{\bf{A}}, \bm{\Sigma}]$ of any order $N \geq 2$. 
\bigskip 

\noindent Part (i):  Let $f \succeq f'$ be as described in Theorem \ref{thm_noisy_attention}, and write $x=\langle f, {\bf{A}}, \bm{\Sigma} \rangle$ and $x'=\langle f', {\bf{A}}, \bm{\Sigma} \rangle$. Normalize $\mu$ to one for ease and for this part only.  

By Lemma \ref{lem_mps}, 
\begin{align*}
&\mathbb{E}_{x'}\left[\exp\left(\nu_{x'}\left(\widetilde{\bm{\omega}}', t\right)\right)\right]\\
=& \sum_{\bm{\omega'}} \mathbb{P}_{x'}\left(\bm{\omega}'\right) \exp\left(\nu_{x'}\left(\bm{\omega}',t\right)\right)\\
=& \sum_{\bm{\omega'}} \mathbb{P}_{x'}\left(\bm{\omega}'\right) \exp\left(\sum_{\bm{\omega}}\pi\left(\bm{\omega}',\bm{\omega}\right)\nu_{x}(\bm{\omega},t)\right)\\
\leq & \sum_{\bm{\omega'}, \bm{\omega}} \mathbb{P}_{x'}\left(\bm{\omega}'\right)\pi\left(\bm{\omega}',\bm{\omega}\right) \exp\left(\nu_{x}(\bm{\omega},t)\right)\\
=& \sum_{\bm{\omega}}\left(\sum_{\bm{\omega}'}\rho(\bm{\omega}'\mid \bm{\omega})\right) \mathbb{P}_x\left(\bm{\omega}\right)\exp\left(\nu_x\left({\bm{\omega}},t\right)\right)\\
=&\mathbb{E}_x\left[\exp\left(\nu_x\left(\widetilde{\bm{\omega}},t\right)\right)\right], 
\end{align*}
and hence \[
\mathbb{E}_{x'}\left[\exp\left(-\nu_{x'}\left(\widetilde{\bm{\omega}}', t\right)\right)\right]\leq \mathbb{E}_x\left[\exp\left(-\nu_x\left(\widetilde{\bm{\omega}},t\right)\right)\right]. 
\] 
Combining these facts with Lemma \ref{lem_noisy_foc} gives the desired result. 

\bigskip 
\noindent Part (ii): Let $f''$ be as described in Theorem \ref{thm_noisy_attention}, and write $x''=\langle f'', {\bf{A}}, \bm{\Sigma}\rangle$. As in Section \ref{sec_proof_eqm}, write 
\begin{equation}
\Delta_{x'', mn}(t)=\mu^{-1}\nu_{x''}\left(\bm{\omega}_{mn},t \right),
\end{equation}
and drop the notation of $t$ in the case of $t=0$. 

By Assumption \ref{assm_u} (iii) and (iv), 
\begin{align*} 
\mathbb{E}_{x''}\left[\exp\left(\widetilde{\Delta}_{x'',mn}(t)\right)\right]
\leq \exp\left(\kappa t/\mu\right) \mathbb{E}_{x''}\left[\exp\left(\widetilde{\Delta}_{x'',mn}\right)\right], 
\end{align*}
where the last term of the above inequality can be bounded above as follows: 
\begin{align*}
\mathbb{E}_{x''}\left[\exp\left(\widetilde{\Delta}_{x'',mn}\right)\right]
=&\sum_{m=1}^K \mathbb{P}_{x''} \left(\bm{\omega}_{mm}\right)\exp\left(\Delta_{x'',mm}\right)+\sum_{m=1}^K \sum_{n=1}^{m-1} \mathbb{P}_{x''} \left(\bm{\omega}_{mm}\right)\gamma\left(\Delta_{x'',mn}\right)\\
\leq & \sum_{m=1}^K \mathbb{P}_{x''} \left(\bm{\omega}_{mm}\right)\cdot 1 +\frac{1}{2}\left(1-\sum_{m=1}^K \mathbb{P}_{x''} \left(\bm{\omega}_{mm}\right)\right)\gamma\left(\Delta_{x'',K1}\right). 
\end{align*}
Thus a necessary condition for $\mathbb{E}_{x''}\left[\exp\left(\widetilde{\Delta}_{x'',mn}\right)\right]\geq 1$ to hold true is
\begin{equation}\label{eqn_necessary}
\Delta_{x'',K1}\geq \gamma^{-1}\left(2\left[ \frac{\exp\left(\kappa \lvert t \rvert/\mu\right) -\sum_{m=1}^K \mathbb{P}_{x''}(\bm{\omega}_{mm})}{1-\sum_{m=1}^K \mathbb{P}_{x''}(\bm{\omega}_{mm})}\right]\right).
\end{equation}

Two things are noteworthy. First, the right-hand side of Condition (\ref{eqn_necessary}) is strictly increasing in $\sum_{m=1}^K \mathbb{P}_{x''}(\bm{\omega}_{mm})$. Second, 
\begin{align*}
\sum_{m=1}^K \mathbb{P}_{x''}\left(\bm{\omega}_{mm}\right)&=\sum_{i,j=1}^N\sum_{m=1}^K f\left(\omega_m \mid a_i\right) f\left(\omega_m \mid a_j\right) \sigma_{ij}
 \geq \sum_{i,j=1}^N \sum_{m=1}^K \frac{1}{K}\cdot \frac{1}{K} \cdot \sigma_{ij}
=\frac{1}{K}, 
\end{align*}
where the inequality can be seen from solving the following optimization problem: \[\min_{z_m, z_m'}\sum_{m=1}^K z_m z_m' \text{ s.t. } z_m, z_m' \geq 0 \text{ } \forall m \text{ and } \sum_{m=1}^Kz_m=\sum_{m=1}^K z_m'=1,\]
which attains its minimum value at $z_{m}=z_m'=\frac{1}{K}$, $m=1,\cdots, K$. 
Thus Condition (\ref{eqn_necessary}) holds true only if 
\begin{equation}\label{eqn_necessary2}
\nu_{x''}\left(\bm{\omega}_{K1},0\right) \geq \mu \gamma^{-1}\left(2\left[\frac{K\exp\left(\kappa \lvert t \rvert/\mu\right) -1}{K-1}\right]\right).
\end{equation}

To complete the proof, notice that the right-hand side of Condition (\ref{eqn_necessary2}) is independent of $x$. By the result of Part (i), the inequality must be strict prior to garbling in order for it to hold true after garbling, and we are done.
\end{proof}

\section{Online Appendix (For Online Publication Only)}\label{sec_online}

\subsection{Costly Information Dissemination}\label{sec_selection}
In this section, suppose the dissemination of political information involves a content provider that profits from voters' eyeballs. By showing a signal $\widetilde{\bm{\omega}}_t$ to voter $t$, the content provider earns a revenue equal to the mutual information between $\widetilde{\bm{\omega}}_t$ and the voter's optimal decision. To maximize revenue, the content provider simply sets $\widetilde{\bm{\omega}}_t=\widetilde{\bf{a}}$ for all $t$. The net profit is then 
\[\int I\left(m_t,\sigma\right)dF(t)-C,\]
where $C>0$ represents the fixed operation cost. 

Two situations can arise in this modified setting: (1) voters pay no attention and act based on the prior, information dissemination is unprofitable and thus unachievable, and candidates adopt their most preferred positions; (2) candidates and voters act as in Section \ref{sec_baseline} and information dissemination incurs no loss, i.e.,
\begin{equation}
\tag{NL}\int I\left(m_t^*,\sigma^*\right)dF(t) \geq C.
\end{equation}
Below we study the equilibria of the second kind.

For any probability matrix $\bm{\Sigma}$, let $H\left(\bm{\Sigma}\right)$ denote the entropy of the policy profile that is distributed according to $\bm{\Sigma}$, and let $\mathcal{E}\left({\bm{\Sigma}}, \mu\right)$ be the set of policy matrices that can be attained in the symmetric equilibria of our interest: 
\begin{align*}
\mathcal{E}\left({\bm{\Sigma}}, \mu\right)=\left\{{\bf{A}}: \exists {\bf{W}} \text{ s.t. } \begin{matrix}  
[{\bf{A}}, {\bm{\Sigma}}] \text{ is } {\bf{W}}-\text{IC}  \\
 {\bf{W}} \text{ is } [{\bf{A}}, {\bm{\Sigma}}]-\text{rationalizable}\\
 \text{(NL)}
\end{matrix} \right\}. 
\end{align*}
The next corollary is immediate: 

\begin{cor}\label{cor_selection}
Assume Assumptions \ref{assm_symmetry} and \ref{assm_u}. Let $\bm{\Sigma}$ be any probability matrix of any order $N \geq 2$ for which the set $\mathcal{E}\left({\bm{\Sigma}}, 0\right)$ is non-empty, and let $C$ be any real number in $\left(0,H\left(\bm{\Sigma}\right)\right)$. As $\mu$ grows from zero to infinity, the set $\mathcal{E}\left({\bm{\Sigma}}, \mu\right)$ satisfies the description of $\mathcal{EA}_t\left(\bm{\Sigma}, \mu\right)$ in Theorem \ref{thm_main}.
\end{cor}

\begin{proof}
The proof is analogous to that of Theorem \ref{thm_main} and is thus omitted.
\end{proof}

\subsection{Noisy News: Effect of Marginal Attention Cost}\label{sec_noisy_mu}
In the model presented in Section \ref{sec_noisy}, write the attention set of any voter $t<0$ as $\mathcal{A}_t\left(\bm{\Sigma}, f, \mu\right)$ in order to make its dependence on $\mu$ explicit. The next corollary generalizes Theorem \ref{thm_main} to encompass noisy news:

\begin{cor}\label{cor_noisy_mu}
Assume Assumptions \ref{assm_symmetry}-\ref{assm_logsupermodular}, and let $\bm{\Sigma}$ be any probability matrix of order $N \geq 2$ for which the set $\mathcal{E}\left({\bm{\Sigma}},f\right)$ is non-empty. Take any $t<0$ such that $\nu_{\langle f, {\bf{A}}, {\bm{\Sigma}}\rangle}\left(\bm{\omega},t\right)>0$ for some ${\bf{A}} \in \mathcal{E}\left({\bm{\Sigma}}, f\right)$ and $\bm{\omega} \in \Omega$. As we increase $\mu$ from zero to infinity, the set $\mathcal{EA}_t\left({\bm{\Sigma}}, f, \mu\right)$ satisfies the description of its equivalent in Theorem \ref{thm_main}. 
\end{cor}

\begin{proof}
Since the set $\mathcal{E}\left(\bm{\Sigma}, f\right)$ is invariant with $\mu$, it suffices to prove the following: as $\mu$ increases  from zero to infinity,
\begin{enumerate}[(i)]
\item the set $\mathcal{A}_t\left(\bm{\Sigma}, f, \mu\right)$ shrinks;
\item the term $\min\left\{u\left(a_1,0\right)-u\left(a_N,0
\right): {\bf{A}} \in \mathcal{A}_t\left({\bm{\Sigma}}, f, \mu\right)\right\}$ increases;
\item the above described objects do not always stay constant.
\end{enumerate}
\bigskip

\noindent Part (i): The proof is analogous to that of Lemma \ref{lem_attentionset} (i) and is thus omitted.
\bigskip 

\noindent Part (ii): Take any policy matrix ${\bf{A}}$ of order $N$  and write $x=\langle f, {\bf{A}}, {\bm{\Sigma}}\rangle$. In the derivation of Condition (\ref{eqn_necessary}), using the fact that $\Delta_{N1}>\Delta_{x}\left(\bm{\omega}_{K1},0\right)$ yields the following necessary condition: 
\[\gamma\left(\Delta_{N1}\right) \geq 2\left[\frac{K\exp(\kappa \lvert t \rvert/ \mu ) -1}{K-1}\right].\]
The remainder of the proof is analogous to that of Lemma \ref{lem_attentionset} (ii) and is thus omitted.

%By Assumption \ref{assm_u} (iii) and (iv), 
%\begin{align*} 
%\mathbb{E}_{x}\left[\exp\left(\widetilde{\Delta}_{x,mn}(t)\right)\right] 
%\leq \exp\left(\kappa t /\mu\right) \mathbb{E}_{x}\left[\exp\left(\widetilde{\Delta}_{x,mn}\right)\right], 
%\end{align*}
%where the last term of the above inequality can be bounded above as follows: 
%\begin{align*}
%\mathbb{E}_{x}\left[\exp\left(\widetilde{\Delta}_{x,mn}\right)\right]
%=&\sum_{m=1}^K \mathbb{P}_{x} \left(\bm{\omega}_{mm}\right)\exp\left(\Delta_{x,mm}\right)+\sum_{m=1}^K \sum_{n=1}^{m-1} \mathbb{P}_{x} \left(\bm{\omega}_{mm}\right)\gamma\left(\Delta_{x,mn}\right)\\
%\leq & \sum_{m=1}^K \mathbb{P}_{x} \left(\bm{\omega}_{mm}\right)+\frac{1}{2}\left(1-\sum_{m=1}^K \mathbb{P}_{x} \left(\bm{\omega}_{mm}\right)\right)\gamma\left(\Delta_{N1}\right),
%\end{align*}
%where \[\Delta_{N1} \triangleq \mu^{-1}\left(u\left(a_1,0\right)-u\left(a_N,0\right)\right).\]
%Thus, a necessary condition for $\mathbb{E}_{x}\left[\exp\left(\widetilde{\Delta}_{x,mn}(t)\right)\right]\geq 1$ to hold true is
%	\[\gamma\left(\Delta_{N1}\right) \geq 2\left[ \frac{\exp(\kappa \lvert t \rvert/ \mu ) -\displaystyle \sum_{m=1}^K \mathbb{P}_{x}(\bm{\omega}_{mm})}{1-\displaystyle \sum_{m=1}^K \mathbb{P}_{x}(\bm{\omega}_{mm})}\right], \] 
%whose right-hand side is bounded below by (the derivation is exact same as that of  (\ref{eqn_necessary2})): \[2\left[\frac{K\exp(\kappa \lvert t \rvert/ \mu ) -1}{K-1}\right].\]
%

\bigskip 
	
\noindent Part (iii): When $\mu$ is small, $\mathcal{A}_t\left({\bm{\Sigma}}, f, \mu\right) \neq \emptyset$ by the assumption that $\nu_{\langle f, {\bf{A}}, {\bm{\Sigma}}\rangle}\left(\bm{\omega},t\right)>0$ for some ${\bf{A}} \in \mathcal{E}\left({\bm{\Sigma}}, f\right)$ and $\bm{\omega} \in \Omega$. When $\mu$ is large, %the following holds true for any policy matrix of order $N$: 
\begin{align*}
\mathbb{E}_{x}\left[\exp\left(\widetilde{\Delta}_{x,mn}(t)\right)\right]
\approx\mathbb{E}_{x}\left[1+\widetilde{\Delta}_{x,mn}(t)\right]\underbrace{<}_\text{(1)} \mathbb{E}_{x}\left[1+\widetilde{\Delta}_{x,mn}\right]\underbrace{=}_\text{(2)} 1,
\end{align*}
where (1) uses Assumption \ref{assm_u} (iii) and (2) Lemma \ref{lem_noisy_symmetry}. 
Thus $\mathcal{A}_t\left(\bm{\Sigma}, f, \mu\right)=\emptyset$ in this case, and this completes the proof.
\end{proof}

\subsection{Limited Commitment}\label{sec_commitment}
In this section, suppose the winning candidate honors his policy proposal with probability $\eta \in [0,1]$ and adopts his most preferred position (assumed to be equal to his type) otherwise. The parameter $\eta$ captures the winner's level of commitment power and is treated as exogenously given. The policy proposal, now regarded as a campaign promise, serves a new role beyond what we have seen so far: it enables voters to infer the winner's type, which is useful in case the latter reneges and does what pleases himself most. 

In what follows, we restrict candidates to using pure symmetric strategies that are strictly increasing in their types. We could alternatively consider mixed strategies that satisfy the monotone likelihood ratio property, whereby high type candidates are more likely to propose high policies than do low type candidates. We choose not to pursue this route for ease of analysis.

%As the reader will soon realize, this assumption helps preserve the monotone relationship between the policy positions and voting  decisions that the analysis so far has utilized a lot. 

Let $-t_N<\cdots<-t_1 \leq 0\leq t_1<\cdots<t_N$ denote the candidates' types and  $-a_N<\cdots<-a_1< 0< a_1<\cdots<a_N$ their policy proposals. The integer $N$ is greater than one and is treated as exogenously given. For $i, j=1,\cdots, N$, write ${\bf{t}}_{ij}=\left(-t_i, t_j\right)$, and define 
\[\widehat{v}({\bf{a}}_{ij}, t)=\eta v\left({\bf{a}}_{ij}, t\right)+(1-\eta)v\left({\bf{t}}_{ij},t\right)\]
as voter $t$'s differential utility from choosing candidate $\beta$ over candidate $\alpha$ under policy profile ${\bf{a}}_{ij}$. Let 
\[\widehat{u}_{+}\left(a_c,t_c\right)=\eta u_{+}\left(a_c,t_c\right)+(1-\eta)u_{+}\left(t_c,t_c\right) \] 
and \[\widehat{u}_{-}\left(a_c,t_{-c}\right)=\eta u_{-}\left(a_c,t_{-c}\right)+(1-\eta)u_{-}\left(t_c, t_{-c}\right)\] be the utility of the winning and losing candidate of type $t_c$ and $t_{-c}$, respectively, when the winning policy is $a_c$. In the formulation of Section \ref{sec_baseline}, replacing $v, u_+$ and $u_-$ with $\widehat{v}, \widehat{u}_+$ and $\widehat{u}_-$, respectively, yields players'  value functions in the current setting. 

Under the aforementioned restrictions, the probability matrix $\bm{\Sigma}$ is fully determined by the candidates' type distribution and will thus be taken as exogenously given. Let $\bf{T}$ be a square matrix of order $N$ whose $ij^{th}$ entry is ${\bf{t}}_{ij}$, and let $\bf{A}$ and $\bf{W}$ be as in the baseline model. A tuple $\langle {\bf{A}}, \bf{W} \rangle$ can be attained in  an equilibrium of our interest if (1) $\bf{A}$ is \emph{incentive compatible for the candidates under $\bf{W}$} (hereinafter, $\bf{W}$-IC), and if (2) $\bf{W}$ is \emph{can be rationalized by optimal attention strategies under $\bf{A}$} (hereinafter, $\bf{A}$-rationalizable). These conditions can be obtained from limiting candidates to adopting monotone pure symmetric strategies and treating $\bf{\Sigma}$ as exogenously given in Definitions \ref{defn_ic} and \ref{defn_rationalizable}. 

For any $\mu>0$, define 
\begin{align*}
\widehat{\mathcal{E}}\left(\mu\right)=\left\{{\bf{A}}: \exists {\bf{W}} \text{ s.t. } \begin{matrix}  
{\bf{A}} \text{ is } {\bf{W}} -\text{IC}  \\
 {\bf{W}} \text{ is } {\bf{A}} -\text{rationalizable}
\end{matrix} \right\}
\end{align*}
as the set of policy matrices that can be attained in the equilibria of our interest, as well as 
\begin{align*}
\widehat{\mathcal{A}}_t\left(\mu\right)=\left\{{\bf{A}}: \mathbb{E}_{[{\bf{A}},{\bf{T}}, {\bm{\Sigma}}]} \left[\exp\left(\mu^{-1} \widehat{v}\left(\widetilde{\bf{a}},t\right)\right)\right] \geq 1\right\}
\end{align*}
as the attention set of any voter $t<0$. Taking intersection yields the set $\widehat{\mathcal{EA}}_t(\mu)$ of policy matrices that captures the voter's attention in the equilibria of our interest.% as well as the set $\widehat{\mathcal{EI}}_t(\mu)$ of policy matrices that induces ideology-based decisions in equilibrium. 

Armed with these definitions and notations, we now state the main result of this section:

\begin{cor}\label{cor_commitment}
Under Assumptions \ref{assm_symmetry} and \ref{assm_u}, the set $\widehat{\mathcal{E}}\left(\mu\right)$ is independent of $\mu$ and is given by 
\[\widehat{\mathcal{E}}=\left\{{\bf{A}}: {\bf{A}} \emph{ is } \widehat{\bf{W}}_N- \emph{IC}\right\}.\]
In the case it is non-empty, the set $\widehat{\mathcal{EA}}_t\left(\mu\right)$ satisfies the description of it equivalent in Theorem \ref{thm_main} for any voter $t<0$ such that $\widehat{v}({\bf{a}}, t)>0$ for some ${\bf{a}} \in {\bf{A}} \in \widehat{\mathcal{E}}$.
\end{cor}

\begin{proof}
Straightforward algebra shows that the term $\widehat{v}\left({\bf{a}}_{ij},t\right)$, which can be elaborated as follows: 
\begin{align*}
\widehat{v}\left({\bf{a}}_{ij},t\right)= & \eta\left[u\left(a_j, t\right)+u\left(-a_j,t\right)-\left[u\left(a_i, t\right)+u\left(-a_i,t\right)\right]\right]\\
& +(1-\eta)\left[u\left(t_j, t\right)+u\left(-t_j,t\right)-\left[u\left(t_i, t\right)+u\left(-t_i,t\right)\right]\right], 
\end{align*}
shares the essential properties of its equivalent $v\left({\bf{a}}_{ij},t\right)$ in the baseline model. Replacing the latter with the former in the proof of Theorem \ref{thm_main} gives the desired result. 
\end{proof}

The impact of limited commitment on equilibrium outcomes can be subtle. Below we illustrate the main intuitions in an example: 

\begin{example}[label=exa:cont3]
In Example \ref{exa:cont1}, suppose the winning candidate honors his policy proposal with probability $\eta$ and reneges with probability $1-\eta$. Tedious but straightforward algebra shows that in this case, the attention set of voter $-\tau$ is defined by the following inequality: 
\begin{equation*}
\eta \cdot (a_2-a_1)+(1-\eta) \cdot \underbrace{(t_e-t_c)}_\text{inference effect} \geq \underbrace{\mu\gamma^{-1}\left(4\exp(2\tau /\mu)-2\right)}_\text{hurdle of indifference}.
\end{equation*}
On the right-hand side of this inequality is the same \emph{hurdle of indifference} that appeared in the baseline model. This term represents the hurdle that candidates need to bypass in order to grab  the voter's attention, and it is this exact term that generates the phenomenon of policy extremism we have seen so far. On the left-hand side of this inequality is a new term called the \emph{inference effect}. It serves as a powerful motivator for the voter to stay attentive, hence inferences can be drawn in case the winner reneges. 

Consider two cases: 
\begin{description}
\item [Case 1] $\mu\gamma^{-1}\left(4\exp(2\tau/\mu)-2\right) \leq t_e-t_c = 1/2$. In this case, the inference effect itself is strong enough to overcome the hurdle of indifference.  As the commitment power increases, the need for drawing inferences diminishes, and the logic behind Theorem \ref{thm_main} kicks in, suggesting that greater policy differentials are needed for arousing and attracting attention. 

\item [Case 2] $ \mu\gamma^{-1}\left(4\exp(2\tau/\mu)-2\right)>t_e-t_c=1/2$. In this case, the inference effect is weak compared to the hurdle of indifference, and the opposite happens as we increase the winner's commitment power. 
\end{description}
\end{example}

\subsection{Multiple Issues}\label{sec_multi}
In this section, suppose a policy consists of two issues $a$ and $b$ (e.g., inflation and unemployment, defense and economy), both taking values in $\Theta$. The Pareto frontier $\mathcal{B}\left(a\right)$, defined as a function of $a$, is strictly decreasing, strictly concave and smooth, and it satisfies the Inada condition $\lim_{a \rightarrow -1} \mathcal{B}'(a)=0$ and $\lim_{a \rightarrow 1} \mathcal{B}'(a)=-\infty$. There is a continuum of voters and two candidates named $a$ and $b$. Each player is either pro-$a$ or pro-$b$, depending on whether his preference weight on issue $b$, or type, belongs to $\Theta_a=\left[-1,0\right]$ or $\Theta_b=\left[0,1\right]$. The type distribution is the same as that of Section \ref{sec_baseline}, subject to minor relabeling. The utility functions $u\left(a,b,t\right)$, $u_{+}\left(a,b, t\right)$ and $u_{-}\left(a,b, t\right)$ are strictly increasing and smooth in $(a,b)$, with $u$ being strictly concave in $(a,b)$ and satisfying the Spence-Mirrlees' generalized single-crossing property: 

\begin{assm}\label{assm_issue}
$u(a,b,t)$ is strictly concave in $(a,b)$ for all $t$ and $-\frac{u_a\left(a,b,t\right)}{u_b\left(a,b,t\right)}$ is strictly increasing in $t$ for all $(a,b)$. 
\end{assm}

Under the above assumptions, each voter $t$'s indifference curve is tangent to the Pareto frontier at a unique point $\left(a^{\circ}(t), \mathcal{B}\left(a^{\circ}(t)\right)\right)$ defined by 
\[\frac{u_a\left(a^{\circ}(t), \mathcal{B}\left(a^{\circ}(t)\right)\right)}{u_b\left(a^{\circ}(t), \mathcal{B}\left(a^{\circ}(t)\right)\right)}+\mathcal{B}'\left(a^{\circ}(t), \mathcal{B}\left(a^{\circ}(t)\right)\right)=0,\]
and $a^{\circ}(t)$ is strictly decreasing in $t$. In the case where the range of $a^{\circ}$, denoted by $a^{\circ}(\Theta)$, coincides with $\Theta$, each $a \in \Theta$ is associated with a unique type $\left(a^{\circ}\right)^{-1}(a)$ of voter whose indifference curve is tangent to the Pareto frontier at $\left(a, \mathcal{B}(a)\right)$.

Since utilities are increasing in both issues, it suffices to consider an augmented economy featuring a single issue $a$ and the following augmented utility functions: $\widehat{u}\left(a,t\right)=u\left(a,\mathcal{B}(a), t\right)$, $\widehat{u}_{+}\left(a,t\right)=u_{+}\left(a,\mathcal{B}(a), t\right)$ and $\widehat{u}_{-}\left(a,t\right)=u_{-}\left(a,\mathcal{B}(a), t\right)$. The resemblance between this economy and the baseline one is noteworthy: 

\begin{lem}\label{lem_issue}
Under Assumption \ref{assm_issue}, the function $\widehat{u}$ satisfies the following properties:
\begin{enumerate}[(i)]
\item $\widehat{u}(\cdot,t)$ is strictly increasing on $\left[-1, a^{\circ}(t)\right]$ and is strictly decreasing on $\left[a^{\circ}(t), 1\right]$ for all $t$;
\item $\widehat{u}(\cdot, t)$ is strictly concave for all $t$;
\item $\widehat{u}$ has strict increasing differences if $u_{at}\geq 0$ and $u_{bt} \leq 0$ and one of these inequalities is strict.
\end{enumerate}
\end{lem}

\begin{cor}\label{cor_issue}
Theorem \ref{thm_main} holds true if the augmented economy satisfies Assumptions \ref{assm_symmetry} and \ref{assm_u}. 
\end{cor}

\bigskip

\noindent Proof of Lemma \ref{lem_issue}
\begin{proof}
Part (i): Differentiating $\widehat{u}(a,t)$ with respect to $a$ and using the fact that \[\frac{u_a\left(a^{\circ}(t), \mathcal{B}\left(a^{\circ}(t)\right), t\right)}{u_b\left(a^{\circ}(t), \mathcal{B}\left(a^{\circ}(t)\right) t\right)}+\mathcal{B}'\left(a^{\circ}(t)\right)=0,\] we obtain
\begin{equation*}
\frac{d}{da}\widehat{u}(a,t)=u_a\left(a,\mathcal{B}(a),t\right)+u_b\left(a, \mathcal{B}(a),t\right)\mathcal{B}'(a), 
\end{equation*}
and hence
\[\frac{d}{da}\widehat{u}(a,t)\bigg\rvert_{a=a^{\circ}(t)}=0.\] Meanwhile, since $a^{\circ}(t)$ is decreasing in $t$ and $a^{\circ}(I)=I$, it follows that for any $a<a^{\circ}(t)$, there exists $t'>t$ such that $a^{\circ}\left(t'\right)=a$ and hence
\[\frac{u_a\left(a, \mathcal{B}\left(a\right), t'\right)}{u_b\left(a, \mathcal{B}\left(a\right),  t'\right)}+\mathcal{B}'\left(a\right)=0.\] Therefore, 
\begin{align*}
\frac{d}{da}\widehat{u}(a,t)\bigg\rvert_{a<a^{\circ}(t)}&=u_a\left(a, \mathcal{B}(a), t\right)-u_b\left(a, \mathcal{B}(a), t\right)\frac{u_a\left(a, \mathcal{B}\left(a\right),  t'\right)}{u_b\left(a,\mathcal{B}(a), t'\right)}\\
&=u_b\left(a, \mathcal{B}(a), t\right)\left[\frac{u_a\left(a, \mathcal{B}\left(a\right),  t\right)}{u_b\left(a,\mathcal{B}(a), t\right)} - \frac{u_a\left(a, \mathcal{B}\left(a\right),  t'\right)}{u_b\left(a,\mathcal{B}(a), t'\right)}\right]\\
&>0, 
\end{align*}
where the inequality follows from the assumption that $u_b>0$ and $-\frac{u_a}{u_b}$ is strictly increasing in $t$. The proof of $\frac{d}{da}\widehat{u}(a,t)\Big\rvert_{a>a^{\circ}(t)}<0$ is analogous and is thus omitted. 

\bigskip

\noindent Part (ii): Straightforward algebra shows that 
%Differentiating $\widehat{u}(a,t)$ with respect to $a$ and rearranging, we obtain 
\begin{align*}
\frac{d^2}{da^2}\widehat{u}(a,t)=&u_{aa}\left(a, \mathcal{B}(a),t\right)+\left[u_{ab}\left(a, \mathcal{B}(a),t\right)+u_{ba}\left(a, \mathcal{B}(a),t\right)\right] \mathcal{B}'(a)\\
&+u_{bb}\left(a, \mathcal{B}(a),t\right)\left(\mathcal{B}'(a)\right)^2+u_b\left(a, \mathcal{B}(a),t\right) \mathcal{B}''(a)\\
=&\left[1, \mathcal{B}'(a)\right]\begin{bmatrix}
u_{aa} & u_{ab}\\
u_{ba} & u_{bb}
\end{bmatrix} \begin{bmatrix}
1\\
\mathcal{B}'(a)
\end{bmatrix} + u_b\left(a, \mathcal{B}(a),t\right) \mathcal{B}''(a)\\
<&0,
\end{align*}
where the inequality follows from the assumption that $u(a,b, t)$ is strictly concave in $(a,b)$ and $\mathcal{B}(a)$ is strictly concave in $a$. 

\bigskip

\noindent Part (iii): Since
\[\frac{\partial^2 \widehat{u}(a,t)}{\partial a \partial t} = u_{at}\left(a,\mathcal{B}(a),  t\right)+u_{bt}\left(a,\mathcal{B}(a),t\right)\mathcal{B}'(a), \] $\frac{\partial^2 \widehat{u}(a,t)}{\partial a \partial t}>0$ if $u_{at}\geq 0$ and $u_{bt} \leq 0$ and one of these inequalities is strict. 
\end{proof}

\end{document}